\documentclass[preprint,12pt]{article}

\usepackage{xcolor} \usepackage{enumerate} \usepackage{subfiles} \usepackage{changepage}
\usepackage[font={small,it}]{caption} \usepackage{verbatim} \usepackage[english]{babel}
\usepackage{url}
\usepackage[utf8]{inputenc} \usepackage[margin=2cm]{geometry}
\usepackage{mathtools}
\usepackage{amsmath, amsfonts, amsthm, amssymb} \usepackage{bm} \usepackage{graphicx} \usepackage{caption}
\usepackage{subcaption} \usepackage{algorithm} \usepackage{algpseudocode} \usepackage{changebar}
\usepackage{tcolorbox}
\usepackage{listings} 
  \usepackage{multirow}
\usepackage{fancyhdr}
\usepackage{titling}
\usepackage{tabularx}
\usepackage{calc}
\usepackage{natbib}
\usepackage{multicol}
\usepackage{caption}
\usepackage{ragged2e}
\usepackage{tikz}
\usepackage{booktabs}
\usepackage{esint}

\newcommand{\er}{Erd\H{o}s-R\'enyi}

\newcommand{\defeq}{\vcentcolon=}
\newcommand{\eqdef}{=\vcentcolon}

\theoremstyle{definition} \newtheorem{definition}{Definition} \newtheorem{theorem}{Theorem}
\newtheorem{proposition}{Proposition} \newtheorem{lemma}{Lemma} \newtheorem{corollary}{Corollary}
\newtheorem{assumption}{Assumption}  \theoremstyle{remark}
\newtheorem*{remark}{Remark}

\makeatletter
\@tfor\next:=abcdefghijklmnopqrstuvwxyzABCDEFGHIJKLMNOPQRSTUVWXYZ\do{ \def\command@factory#1{
\expandafter\def\csname bb#1\endcsname{\mathbb{#1}} \expandafter\def\csname
cl#1\endcsname{\mathcal{#1}} \expandafter\def\csname bf#1\endcsname{\mathbf{#1}} }
\expandafter\command@factory\next } 

\newcommand{\epsX}{\epsilon_X}
\newcommand{\epsY}{\epsilon_Y}
\DeclareMathOperator*{\argmax}{arg\,max}

\title{Convergence and Connectivity:  Dynamics of Multi-Agent Q-Learning in Random Networks}

\makeatletter
\def\@maketitle{%
  \newpage
  \vspace*{2em}
  \begin{center}%
    \let\footnote\thanks
    {\Large \bfseries\@title \par}%
    \vskip 2em%

    {\small
      \begin{tabular}{@{}p{0.3\textwidth}p{0.36\textwidth}p{0.3\textwidth}@{}}
      \raggedright
      \textbf{Dan Leonte}\par
      KAUST\par
      Saudi Arabia\par
      \texttt{dan.leonte@kaust.edu.sa}
      &
      \raggedright
      \textbf{Aamal Hussain}\par
      Latent Technology\par
      London, UK\par
      \texttt{aamal@latent-technology.com}
      &
      \raggedright
      \textbf{Rapha\"el Huser}\par
      KAUST\par
      Saudi Arabia\par
      \texttt{raphael.huser@kaust.edu.sa}
    \end{tabular}
    \vskip 0.5em
    \begin{tabular}{@{}p{0.45\textwidth}p{0.45\textwidth}@{}}
      \raggedright
      \textbf{Francesco Belardinelli}\par
      Imperial College London\par
      London, UK\par
      \texttt{francesco.belardinelli@imperial.ac.uk}
      &
      \raggedright
      \textbf{Dario Paccagnan}\par
      Imperial College London\par
      London, UK\par
      \texttt{d.paccagnan@imperial.ac.uk}
    \end{tabular}
    }%
  \end{center}%
  \par
  \vskip 1.5em
}

\begin{document}

\maketitle

\begin{abstract}
  Beyond specific settings, many multi-agent learning algorithms fail to converge to an equilibrium solution, instead displaying complex, non-stationary behaviours such as recurrent or chaotic orbits. In fact, recent literature suggests that such complex behaviours are likely to occur when the number of agents increases.
In this paper, we study Q-learning dynamics in network polymatrix normal-form games where the network structure is drawn from classical random graph models.  In particular, we focus on the Erdős-Rényi model, which is used to analyze connectivity in distributed systems, and the Stochastic Block model, which generalizes the above by accounting for community structures that naturally arise in multi-agent systems. 
In each setting, we establish sufficient conditions under which the agents' joint strategies converge to a unique equilibrium. We investigate how this condition depends on the exploration rates, payoff matrices and, crucially, the probabilities of interaction between network agents. 
We validate our theoretical findings through numerical simulations and demonstrate that convergence can be reliably achieved in many-agent systems, provided interactions in the network are controlled.
\end{abstract}


\section{Introduction}
\label{sec::Introduction}
The development of algorithms for multi-agent 
learning has produced a number of successes in
recent years, solving challenging problems in load-balancing
\citep{krichene:routing,southwell:graphical-congestion}, energy management
\citep{khaled:energy-ga} and game playing
\citep{moravcik:poker,brown:poker,foerster:smac,tuyls:stratego}.
Also thanks to these successes, the game-theoretic foundations of
learning in the face of many agents remain a thriving area of research. As the number of agents grows in these systems, it is critical to understand if their learning algorithms are guaranteed to converge to an equilibrium solution, such as a Nash Equilibrium.


 Unfortunately, previous work suggests that non-convergent behaviour becomes the norm as the number of players increases. A strong example is \citep{piliouras:impossibility}, which 
introduced a game in which no independent learning
dynamics converges to a Nash Equilibrium. 
Further studies have shown that chaotic dynamics may
occur even in two-player finite-action games \citep{sato:rps,galla:complex}. Crucially, 
\cite{sanders:chaos} found that, as the number of players 
grows, non-convergent
behaviour becomes the norm. 

Whilst at first glance these results suggest an insurmountable barrier towards strong
convergence guarantees, a key missing factor is an in-depth analysis of the \emph{interactions} between agents. In particular, both \cite{sanders:chaos} and \cite{hussain:aamas} assume
that the payoff to any given agent is dependent on \emph{all} other agents in the environment. This assumption rarely holds in practice. Rather, agents often interact through an underlying
network that may represent communication constraints or spatial proximity. A practitioner or system designer has a certain degree of control over the network structure, e.g., in a robotic swarm, they can influence the network by adjusting the communication range of the robots, while in a sensor network, this can be done by controlling the density of the deployed sensors. This leads us to study the following research question:

\begin{center}
    \textit{How does network structure affect the convergence of learning as the number of agents increases?}
\end{center}

Answering this question represents a key step towards guaranteeing the feasibility of learning with many
agents, so long as the network structure can be controlled. Indeed, a number of works have examined
learning in network games, uncovering the relationship between the network structure and properties
of the equilibrium as well as
designing algorithms that converge to an equilibrium
\citep{melo:network,parise:network,melo:qre,Shokri2020Leader-FollowerActiveness}. However, many of these algorithms have unrealistic requirements, such as full knowledge of the agents' payoff functions or their gradients.
Our goal instead is to
consider a widely applied reinforcement learning algorithm -- \emph{ Q-Learning} -- which requires only evaluations of the payoff functions, and study how the parameters of the
algorithm and the structure of the underlying network can be leveraged to yield convergence
guarantees, as well as how such guarantees scale with the number of agents.

\paragraph{Model.} We study agents who update their strategies via the Q-Learning
dynamics \citep{sato:qlearning}, a continuous-time counterpart to the well-studied
Q-Learning algorithm \citep{Watkins:qlearning,sutton:barto}. We focus our study on
\emph{network polymatrix games} \citep{janovskaja:polymatrix} in which agents select from a finite set of actions and interact with their neighbours on an underlying network. Network polymatrix games provide a setting in which the topology of interaction directly shapes both payoff structure and learning dynamics, allowing us to isolate and quantify the role of network structure in the convergence of multi-agent online learning. We study networks which are drawn from \emph{random graph} models:
the {\em \er} model \citep{erdos-renyi} and
the {\em Stochastic Block 
model} \citep{holland:sbm}. The \er\ model is widely used to study communication in distributed systems \citep{hong:oco-er,saurav:distributed-er}. The Stochastic
Block model is a natural extension of \er\ to model community structures which naturally arise in distributed systems \citep{prioutiere:clustering-sbm}. 


To prove our main results (Theorems \ref{thm::er-convergence} and \ref{thm::sbm-convergence}), we restrict to network polymatrix games in which each edge of the network is assigned the same bimatrix game. This means that agents are engaged in the same interaction with multiple opponents. This assumption is commonly
used to study the emergence of cooperation \citep{zhang:pd-graphs,chakraborty:chaos} and congestion
\citep{szabo:congestion} in multi-agent systems. Indeed, the assumption of shared payoffs has led to
a number of successes in studying large scale systems with many agents
\citep{perrin:fp,zida:population,perrin:omd,parise:graphon,hu:ql-mfg}. We adopt this assumption, but
also empirically validate our claims outside of this framework through the \emph{Conflict} Network game \citep{ewerhart:fp}, in which payoff matrices vary across edges.

\paragraph{Contributions.}
Our main contribution is to characterise the convergence properties of Q-Learning, depending on the exploration rates of the agents, the game payoffs, and the 
expected degree of the network. 
Specifically, our results provide sufficient conditions for the Q-learning dynamics to converge to an equilibrium with high probability as the number of agents $N$ increases, provided that the
network is drawn from one of the two above-mentioned models. 
We demonstrate that in low-degree networks, Q-learning converges even with low exploration rates and with a large number of agents. In fact, our experiments show convergence even with $200$ agents in low-degree networks, whereas \citet{sanders:chaos} reported failure to converge with low exploration rates for as few as five agents.

Our work establishes an explicit relationship between the convergence
of Q-Learning Dynamics and the expected node degree in the network, thus ensuring the feasibility of learning in many-agent games, so long as the expected node degree is controlled. Our results further ensure the uniqueness of the equilibrium, meaning that Q-learning converges to a single solution regardless of the initial conditions. To the best of our knowledge, this is the first work to study the asymptotic behaviour of learning dynamics in the context of network polymatrix games with random graph models, and to derive the relation between convergence and the expected node degree in the network. 
For reasons of space, the proofs of all results appear in the Appendix. 
%
%
\subsection{Related Work}
Our work focuses on independent online learning in network polymatrix games. Network games are
well-studied in the setting of zero-sum games \citep{cai:minimax}, which model strictly
competitive systems. In such cases, it is known that the continuous-time counterparts of popular
algorithms such as Fictitious Play \citep{ewerhart:fp} and Q-Learning \citep{piliouras:zerosum}
converge to an equilibrium. By contrast, \cite{piliouras:hamiltonian} showed that the \emph{Replicator
Dynamics} \citep{smith:replicator}, a continuous-time model of the Multiplicative Weights Update
algorithm \citep{hazan:mwu}, does not converge to the Nash Equilibrium. Similarly, \citet{shapley:twoperson} showed the non-convergence of Fictitious Play in a two-person non-zero-sum game. Indeed, \citet{hart:uncoupled} showed that \emph{no} learning dynamic that is uncoupled, in that the strategy of a player is independent of the payoff functions of other players, can converge to a Nash Equilibrium in all normal-form games. \citet{piliouras:impossibility} extended this result to show the impossibility of convergence of a wide class of learning dynamics to approximate Nash-Equilibria in normal-form games.

Network games have also been studied to understand the properties, in particular the uniqueness, of
the equilibrium \citep{bramoulle:networks,parise:network,melo:qre}. In many cases, the
literature appeals to the study of \emph{monotone games} \citep{paccagnan2018nash} which subsumes
zero-sum network games \citep{akin:zero-sum}.
The formalism of monotone games has been applied to design algorithms that provably converge to Nash
Equilibria. However, many of these algorithms require that agents have full access to their payoff function
\citep{parise:network} or its gradient \citep{mertikopoulos:concave}. Monotone
games also share strong links with the idea of \emph{payoff perturbations} \citep{facchinei:VI} in
which a strongly convex penalty is imposed to agents' payoff functions to stabilise learning
\citep{abe:perturbed-md,sokota:qre,liu:regularisation}.

Our work departs from the above by considering \emph{Q-Learning Dynamics}, a foundational exploration-exploitation model central to reinforcement learning
\citep{albrecht:marl,tuyls:foundational-models}. Q-Learning Dynamics is also related to the replicator dynamics \citep{bloembergen:review} and to Follow-the-Regularised-Leader
\citep{abe:MFTRL}. Outside of specialised classes of games, Q-Learning Dynamics has been shown to exhibit chaotic orbits \citep{sato:rps}, a phenomenon which becomes more
prevalent as the number of players increases \citep{sanders:chaos}. Similar to our work,
\cite{hussain:aamas-two} study the convergence of Q-Learning in network polymatrix games. However, their work considered deterministic graphs and, as such, derived sufficient conditions for convergence in terms of norms of the network adjacency matrix. We instead consider Q-Learning in a stochastic setting, where stochasticity arises from the random network. In doing so, we derive a direct relationship between the expected node degree and the convergence of Q-Learning Dynamics. Our random-network analysis also allows for the study of certain network classes, e.g., the Stochastic Block Model, which are well-studied models of distributed systems that contain sub-communities. With the tools of random networks, we make inroads into understanding how 
inter-, or intra- community edge connections affect the resulting convergence of the learning algorithm. 

Other works in multi-agent learning and online optimisation have considered the random network setting, for example to determine the existence of pure Nash Equilibria \citep{daskalakis:random-games}, to study the performance of distributed online algorithms \citep{hong:oco-er}, or to examine the emergence of cooperative behaviors \citep{marsili:collaboration}. A parallel step in the study of network games was the introduction of \emph{graphon games} in \citep{parise:graphon,carmona:gmfg}. The authors consider network games in the limit of uncountably infinite agents and generalise the \er\ and Stochastic Block models. This also extends mean-field games \citep{lauriere:survey,hu:ql-mfg} by introducing heterogeneity amongst players through their edge connections. Subsequent works on graphon games largely focus on the analysis of equilibria \citep{caines:gmfg,carmona:lqr-gmfg} or design learning algorithms that converge in time-average to an equilibrium \citep{koeppl:approx-gmfg,zhang:regularised-gmfg,zhang:unknown-gmfg}. By contrast, our goal is to understand the \emph{last-iterate} behaviour of Q-Learning Dynamics and establish probabilistic bounds to guarantee convergence in games with finitely many players. In particular, we study how the probability of edge connections in the network curtails non-convergent dynamics as the number of players increases. Thus, the restriction to uncountably infinite players is not suitable for our purposes, although the analysis of non-convergent learning algorithms in graphon games is an interesting direction for future work. 

\section{Background}
\label{sec::prelims}
\paragraph{Game Model.}\label{sec::game-model} 
We consider \emph{network polymatrix games}
\citep{janovskaja:polymatrix,cai:minimax},
which are defined as tuples $\clG = (\clN, \clE,
	(\clA_k)_{k \in \clN}, (A^{kl}, A^{lk})_{(k, l) \in \clE})$, where $\clN = \{1, 2, \ldots, N \}$
denotes a set of $N$ agents, indexed by $k$. The interactions between agents are modelled by a set $\clE \subset \clN \times \clN$ of
edges that defines an \emph{undirected network}. An alternative formulation of the underlying network is through an \emph{adjacency matrix} $G \in
\bbR^{N \times N}$ in which $[G]_{kl}=1$ if $(k, l) \in \clE$, and $[G]_{kl}=0$ otherwise. The \emph{degree} of an agent $k \in \clN$ is the number of edges in $\clE$ that include $k$.

At each round, each agent $k$ selects an action $i \in \clA_k$, where $\clA_k$ is a finite set of
$n_k$ actions. We denote the \emph{strategy} of an agent as the
probability distribution over their actions. Next, we define the \emph{joint strategy} across all agents as the concatenation of all individual strategies $\bfx = (\bfx_1, \ldots, \bfx_N)$ and apply the shorthand $\bfx_{-k}$ to denote the strategies of all agents other than $k$.

The goal of each agent is to maximise a utility function $u_k$. In a network polymatrix game $\clG$, each edge is associated with the \emph{payoff matrices} $(A^{kl},
A^{lk})$, i.e., the payoff for each agent $k \in \mathcal{N}$ is  
\begin{equation*}
	u_k(\bfx_k, \bfx_{-k}) = \sum_{l : (k, l) \in \clE} \bfx_k^\top A^{kl} \bfx_l.
\end{equation*}
\paragraph{Solution Concepts.} \label{sec::solution-concepts}
We focus on two widely-studied solution concepts: the {\em Nash Equilibrium}
\citep{Nash1950EquilibriumGames} and {\em Quantal Response Equilibrium} \citep{mckelvey:qre}. To
define these concepts we first introduce the \emph{reward} to agent $k \in \clN$ for playing action $i
\in \clA_k$ as $r_{ki}(\bfx_{-k}) = \sum_{l : (k, l) \in \clE} \sum_{j \in \clA_l} [A^{kl}]_{ij} x_{lj}$.
\begin{definition}[NE]
	A joint strategy $\bfx^* \in \Delta(\clA)$ is a {\em Nash Equilibrium} if
	\begin{equation*}
		\bfx_k^* \in \argmax_{\bfx_k \in \Delta(\clA_k)} u_k(\bfx_k, \bfx_{-k}^*) \text{ for all } k \in \clN. 
	\end{equation*} 
\end{definition}
\begin{definition}[QRE]
	Let $T_1, \ldots, T_N > 0$. Then, a joint strategy $\bfx^* \in \Delta(\clA)$ is a {\em  Quantal
	Response Equilibrium} if 
	\begin{equation*}
		\bfx_k^* =  \frac{\exp\left(
			r_{ki}(\bfx_{-k}^*) / T_k
			\right)}{ \sum_{j \in \clA_k} \exp\left( r_{kj}(\bfx_{-k}^*) / T_k \right)} \text{ for all } k \in \clN.
	\end{equation*}
\end{definition}
	The QRE is a natural extension of the Nash Equilibrium that allows agents to play suboptimal
	actions with non-zero probability. This is crucial in online learning, where agents must
	\emph{explore} their strategy space. The parameter $T_k$ is therefore known as the
	\emph{exploration rate} of agent $k$. Notice that, by taking the limit $T_k \rightarrow 0$ for
	all $k$, the QRE converges to the Nash Equilibrium \citep{mckelvey:qre}.


	\paragraph{Q-Learning Dynamics.} \label{sec::q-learning}
	We now describe the Q-Learning algorithm \citep{Watkins:qlearning,sutton:barto} which aims to learn an optimal action-value function $Q : \clA \rightarrow \bbR$ that
	captures the expected reward of taking a given action. To do this, an agent must \emph{explore} their action space by playing possibly suboptimal actions in an attempt to discover optimal strategies. Simultaneously, the agent plays optimal actions based on their current Q-estimate. As such, the Q-learning algorithm is a foundational model to describe the behaviour of agents that must balance exploration and exploitation. We consider the multi-agent extension of
	Q-Learning \citep{schwartz:marl,schaefer:marl} in which each agent $k$ maintains an individual
	Q-value estimate $Q_{ki} \in \bbR$ for each of their actions $i \in \clA_k$. These are updated at each round $t$ via the update
	\begin{equation} \label{eqn::qld_update}
		Q_{ki}(t + 1) = (1 - \alpha_k) Q_{ki}(t) + \alpha_k r_{ki}(\bfx_{-k}(t)),
	\end{equation}
	where $\alpha_k \in (0, 1)$ is the learning rate of agent $k$. 
    
    In effect, $Q_{ki}$ gives a
	discounted history of the rewards received when action $i$ is played, with $1 - \alpha_k$ as the
	discount factor. Note that Q-values are updated by the rewards
	$r_{ki}(\bfx_{-k})$ that depend on the strategies of all other agents at time $t$. So, the
	reward for a single action can vary between rounds. This non-stationarity often leads
	to chaotic dynamics in multi-agent learning.
	
    Given the Q-values, each agent plays their actions according to the \emph{Boltzmann}
	distribution:
	\begin{equation} \label{eqn::boltzmann-selection}
		x_{ki}(t + 1) = \frac{\exp \left( Q_{ki}(t + 1) / T_k \right)}{\sum_{j \in \clA_k} \exp \left( Q_{kj}(t + 1) / T_k \right)},
	\end{equation}
	where $T_k \in (0, \infty)$ is the exploration rate of agent $k$. 
    
	\cite{tuyls:qlearning} and \cite{sato:qlearning} showed that a continuous-time approximation of the Q-Learning
	algorithm is given by a variation of the replicator dynamics \citep{smith:replicator,hofbauer:egd} that is called the \emph{Q-Learning Dynamics}
    
	{\small
	\begin{equation} \tag{QLD} \label{eqn::QLD}
		\frac{\dot{x}_{k i}}{x_{k i}}=r_{k i}\left(\mathbf{x}_{-k}\right)-\langle \mathbf{x}_k, r_k(\mathbf
		{x}) \rangle +T_k \sum_{j \in \clA_k} x_{k j} \ln \frac{x_{k j}}{x_{k i}}.
	\end{equation}
	}
where $\langle \cdot, \cdot \rangle$ denotes the scalar product.
\cite{piliouras:zerosum} proved that the fixed points of (\ref{eqn::QLD}) coincide with the QRE of the game.


\paragraph{Additional Notation.}
Given a square matrix $A \in \bbR^{N \times N}$, we denote its {\em spectral norm} as $\lVert A \rVert_2 = \sup_{x \in
\bbR^N \, : \, \lVert \bfx \rVert_2 = 1} \lVert Ax \rVert_2$. If $A$ is symmetric, all its eigenvalues $\lambda_1, \ldots, \lambda_N$ are real, and the spectral norm agrees with the {\em spectral radius}, which is defined as
$\rho(G) = \max\{\lambda_1,\ldots,\lambda_N\}.$
\section{Convergence of Q-Learning in Graphs}
\label{sec::erdos-renyi}

In this section, we show that the convergence of (\ref{eqn::QLD}) is closely related to the structure of the underlying graph. First, we present the problem setup: we specify the payoffs by the \emph{intensity of identical interests} framework \citep{hussain:aamas-two} and specialize to a certain class of network polymatrix games, which satisfy Assumption \ref{ass::bimatrix-network} below. Second, we establish in Lemma \ref{lem::qld-spectral-radius} a sufficient condition on the exploration rates $T_k$ such that (\ref{eqn::QLD}) converges to the \emph{unique} QRE. We adapt this result to the random network setup using the \er\ and Stochastic Block models. In both cases, we establish connections between (\ref{eqn::QLD}) and the expected degree of a node in the network.
\begin{definition}[Intensity of Identical Interests]
	Let $\clG = (\clN, \clE, (\clA_k)_{k \in \clN}, (A^{kl}, A^{lk})_{(k, l) \in \clE})$ be a
	network polymatrix game. Then the intensity of identical interests $\delta_I$ of $\clG$ is given
	by
	\begin{equation} \label{eqn::intensity-identical-interests}
		\delta_I = \max_{(k, l) \in \clE} \lVert A^{kl} + (A^{lk})^\top \rVert_2.
	\end{equation}
    \label{def:id_interests}
\end{definition}

The intensity of identical interests $\delta_I > 0$ measures the similarity between the payoffs of connected
agents across all edges. A canonical example is the \emph{pairwise zero-sum game} in which $A^{kl} = -(A^{lk})^\top$ for all $(k, l) \in \clE$. In this case, the intensity of identical	interests is zero.

We prove our main results under the following Assumption on network polymatrix games.

\begin{assumption} \label{ass::bimatrix-network}
	Each edge is assigned the same bimatrix game, i.e., $(A^{kl}, A^{lk}) = (A, B)$ $\forall (k, l) \in \clE$.
\end{assumption} 

Note that Assumption \ref{ass::bimatrix-network} does not require $A^{kl} = A^{lk}$, rather it requires that each edge is associated with the same \emph{pair} of matrices $(A, B)$. No assumption is made about which agent receives matrix $A$ and which receives matrix $B$. This is reflected in our experiments (Section \ref{sec::experiments}) in which payoff matrices are randomly assigned to agents on each edge. This aspect of the setting is discussed in detail in Appendix \ref{app::undirected}. In our experiments, we also validate that our claims outside of this Assumption by studying the Conflict Network Game \citep{ewerhart:fp}; see Section \ref{sec::experiments} for details.

Having specified our setting, we next determine a sufficient condition for the convergence of
(\ref{eqn::QLD}) in terms of the properties of the adjacency matrix $G$. All proofs are deferred to
Appendix \ref{app::proof1} and \ref{app::proof2}.

\begin{lemma}\label{lem::qld-spectral-radius}
	Let $\clG = (\clN, \clE, (\clA)_k, (A, B)_{(k, l) \in \clE})$ be a network polymatrix game that satisfies Assumption \ref{ass::bimatrix-network}, and let $G$ be the adjacency matrix associated with the edge-set $\clE$. If for all agents  $k \in \clN$, $T_k > \delta_I \cdot \rho(G)$, 
	then the QRE $\bfx^*$ of the game $\clG$ is unique. 
    Further, $\bfx^*$ is globally asymptotically stable under (\ref{eqn::QLD}).
\end{lemma}
This result has two components: (i) the QRE is unique, and (ii) it is asymptotically stable with respect to (\ref{eqn::QLD}). Together, these statements ensure that a unique equilibrium solution is learned.

\paragraph{Discussion.}
Recall from Section \ref{sec::solution-concepts} that lower exploration rates are preferred, so that the equilibrium solution remains close to a Nash Equilibrium. Lemma \ref{lem::qld-spectral-radius} provides two quantities, $\delta_I$ and $\rho(G)$ which can be chosen appropriately to guarantee the convergence of (\ref{eqn::QLD}).  A common approach to address this is to fix the game to a specific class. As an example, restricting to pairwise zero-sum games, where $\delta_I = 0$, ensures that (\ref{eqn::QLD}) converges as long as $T_k > 0$ for all $k \in \clN$, regardless of $\rho(G)$. Our approach, in contrast, focuses on controlling the spectral radius $\rho(G)$, while treating $\delta_I$ to be fixed a-priori. 
To achieve this, we explore how network sparsity -- a tunable parameter in decentralised systems such as robotic swarms or sensor networks -- impacts the convergence of (\ref{eqn::QLD}). Our analysis turns to \emph{random networks}, where network sparsity is directly parameterised through the probability of edge connections. 
By deriving probabilistic upper bounds on the (random) spectral radius term $\rho(G)$ in terms of these parameters, we can directly link the network's sparsity to the convergence or (\ref{eqn::QLD}), providing robust convergence guarantees that hold for broad classes of networks.



We specify the network structure by
the parameters of the model from which it is generated. In the \er\ model, this is the
probability $p$ that a pair of nodes is connected by an edge. In the  Stochastic Block model, it is the set of within-community edge probabilities $p_1,\ldots,p_C$ and between-community edge probability $q$. We stress that our theoretical bounds are not limited to these specific models: they apply to any random network model with independent entries (see Lemma~\ref{lemma_we_need}). 

\paragraph{\er.} We begin with the \er\ (ER) model \citep{erdos-renyi}, in which the graph is generated by independently sampling each edge with probability $p \in (0, 1)$. The adjacency matrix $G$ is then a random matrix with entries that are Bernoulli distributed with parameter $p$. Let the degree of a node $k$ be the number of edges $(k, l) \in
\clE$ that contain $k$. The expected node degree in a network drawn from the ER model is therefore $p(N - 1)$. 

\paragraph{Stochastic Block Model.}  Networks are drawn from the Stochastic Block (SB) model \citep{holland:sbm} by partitioning the nodes into $C$
communities. The probability of an edge between nodes $k$ and $l$ that are in the same community
$c$ is given by $p_c \in (0, 1)$, whereas the probability of an edge between nodes in different
communities is $q \in (0, 1)$. The adjacency matrix $G$ is then generated by sampling each
edge independently according to the Bernoulli distribution with probabilities $p_c$ and $q$. The expected node degree in community $c$ is $p_c (N_c - 1) + q(N - N_c)$, where $N_c$ is the number of nodes in community $c$. Note that the higher $p_c$ and $q$, the higher the expected node degree, and further that the ER model is a particular case of the SB model for a single
community, i.e., $N_c = N$.

We first place probabilistic bounds on the spectral radius of the adjacency matrices in the ER and SB models, which we then use in Theorems \ref{thm::er-convergence} and \ref{thm::sbm-convergence} to obtain the main convergence results. For simplicity, in the SB model case
, the
following results assume all communities are of equal size, specifically $N_c = N / C$ for all
$c$. However, the proof stays valid when the communities' sizes vary.
\begin{lemma}\label{lem::er-spectral-radius}
	Consider the symmetric adjacency matrix $G \in \bbR^{N \times N}$ of a network drawn from the ER model, i.e., $G$ has a zero main-diagonal and the off-diagonal entries $\{ [G]_{kl} \sim \texttt{Bernoulli}(p): 1 \leq k < l \leq N\}$ are independent, where $p > 0$. 
    For any $\epsilon > 0$, it holds with probability greater than $1 - \epsilon$, 
	\begin{equation} \label{eqn::er-cond}
    		\rho(G) \le \rho^{ER} \defeq  (N-1)p + \sqrt{2(N-1)\, p \, (1-p) \log{\frac{2N}{\epsilon}}  } + \frac{2}{3} \log{\frac{2N}{\epsilon}}.
	\end{equation}
\end{lemma} 
\begin{lemma}\label{lem::sbm-spectral-radius}
	Consider the symmetric adjacency matrix $G \in \bbR^{N \times N}$ of a network drawn from the SB model
	\[
		\begin{cases}
			[G]_{kl}  \sim \texttt{Bernoulli}(p_c), & \text{ if } k \text{ and } l \in c,                     \\
			[G]_{kl}  \sim \texttt{Bernoulli}(q),   & \text{ if } k \in c \text{ and } l \in c^\prime \neq c, \\
			[G]_{kk} = 0,                           & \text{ for } 1 \le k \le N,
		\end{cases}
	\]
	where $\{[G]_{kl} : 1 \leq k < l \leq N \}$ are independent. For any $\epsilon > 0$, it holds with probability greater than $1-\epsilon$
    \begin{align*}
        \rho(G) \leq \rho^{SB} \eqdef & (N-N/C)q + (N/C-1)p_{\text{max}}  + \\
     + & \sqrt{2 \left( (N-N/C)  q (1-q) + (N/C-1) \sigma_{\text{p,max}}^2 \right) \log{\frac{2N}{\epsilon}}} + \frac{2}{3} \log{\frac{2N}{\epsilon}}.
    \end{align*}
	where $p_{\max} = \max \{ p_1, \ldots, p_C \}$ and $\sigma_{p,\max} = \max \{\sqrt{p_1(1 - p_1)}, \ldots, \sqrt{p_C(1 - p_C)}\}$.
\end{lemma}
Together with Lemma \ref{lem::qld-spectral-radius}, the two results above allow us to establish a link between the edge connection probabilities in the graph -- and thus the expected node degree -- and the convergence of (\ref{eqn::QLD}).

\begin{theorem} \label{thm::er-convergence} Let $\epsilon > 0$ and $\clG$ a network
	polymatrix game satisfying Assumption \ref{ass::bimatrix-network} with corresponding network drawn from the ER model with parameter $p$. If
    $$T_k(N) > \rho^{ER}
    \text{  for all agents } k \in \clN,$$
    then with probability at least $1-\epsilon$, there is a unique QRE $\bfx^* \in \Delta$ and, for any initial condition, the (\ref{eqn::QLD}) trajectories converge to $\bfx^*$.
\end{theorem}
\begin{theorem} \label{thm::sbm-convergence}
	Let $\epsilon > 0$ and $\clG$ a network polymatrix game satisfying Assumption
	\ref{ass::bimatrix-network} with corresponding network drawn from the SB model with probabilities $p_1,\ldots,p_C, q$.
    If 
    \[ T_k(N) > \rho^{SB} 
        \text{ for all agents } k \in \clN , \]
then with probability at least $1-\epsilon$, there is a unique QRE $\bfx^* \in \Delta$ and, for any initial condition, the (\ref{eqn::QLD}) trajectories converge to $\bfx^*$.
\end{theorem}
\begin{remark}
    In both the ER and SB model, the dominant $O(N)$ terms in the spectral-radius bounds are governed by expected node degrees. Indeed, the leading term in $\rho^{ER}$, namely $(N-1)p$, equals the expected degree in the ER model. Likewise, the leading term in $\rho^{SB}$, namely $$(N-N/C)q + (N/C-1) p_{\max},$$ represents the maximum expected degree across all blocks in the SB network. Further, the sub-leading $O(\sqrt{N \log N})$ term is governed by the variances of the degree distributions: $p(1-p)$ in the ER case, and on the block-wise maximum variance $\sigma_{p,\max}$ and variance $q(1-q)$ in the SBM case. 
\end{remark}
    \textbf{Implications to learning in multi-agent games.} Theorems \ref{thm::er-convergence} and \ref{thm::sbm-convergence} show that (\ref{eqn::QLD}) converges with low exploration rates whenever the expected degree of nodes in the network is small. Importantly, previous works \citep{sanders:chaos} linked the onset of unstable dynamics to the total number of agents $N$, but only in the limited case where all players interact with one-another, i.e., $p=1$. Our bounds refine this finding by showing that the onset of instability is related to the expected degree in the network, not to $N$ explicitly.
    Further, once we control the expected degree, the variance terms are automatically controlled as well. As a consequence, the threshold for chaotic behavior increases from linear to exponential in $N$. In other words, with controlled expected degrees, chaotic dynamics can only emerge at exponentially large network sizes, far beyond the regimes considered in earlier analyses.

\section{Experimental Evaluation}
\label{sec::experiments}
The aim of this section is to illustrate the implications of Theorems \ref{thm::er-convergence} and
\ref{thm::sbm-convergence}. 
We simulate the Q-Learning algorithm 
on network games, where the network is drawn from the ER and SB models. Note that we simulate the discrete-time algorithm described in Section~\ref{sec::q-learning}, rather than integrating the continuous-time model (\ref{eqn::QLD}). This is to explore the agreement between the theoretical predictions and the discrete-time update, which is the algorithm that is applied in practice. Each experiment runs for 4000 iterations, and we assess convergence numerically, as explained in Appendix \ref{appendix:simulation_setup}. We show that, as the expected node degree of the network increases, higher exploration rates are required for Q-Learning to converge by finding the boundary between stable and
unstable dynamics. We analyse how
the boundary evolves with $N$ and provide empirical evidence that our results 
extend beyond the assumptions used to derive the 
bounds.
    In our experiments, we assume all agents have the same exploration rate $T$ and thus omit the dependency on $k$. As our results are lower bounds across all $k$, this is without loss of generality.
\paragraph{Game Models.} We first simulate Q-Learning in the Network Shapley and Network Sato games, which are extensions of
classical bimatrix games, introduced by \cite{shapley:twoperson} and \cite{sato:rps}
respectively, to the network polymatrix setting. In both cases, each edge defines the same bimatrix
game, $(A^{kl}, A^{lk}) = (A, B) \; \forall (k, l) \in \clE$ 
where in the Network Shapley game:
\begin{eqnarray*} 
    A=\begin{pmatrix} 1 & 0 & \beta \\
                \beta & 1 & 0 \\ 0 & \beta & 1 \end{pmatrix}, \, B=\begin{pmatrix} -\beta & 1 & 0 \\ 0 &
                -\beta & 1 \\ 1 & 0 & -\beta \end{pmatrix}, 
    \end{eqnarray*}
for $\beta \in (0, 1)$; and in the Network Sato game:  
\begin{eqnarray*} 
    A=\begin{pmatrix} \epsX & -1 & 1 \\
    1 & \epsX & -1 \\
    -1 & 1 & \epsX \end{pmatrix}, \, B=\begin{pmatrix} \epsY & -1 & 1 \\
    1 & \epsY & -1 \\
    -1 & 1 & \epsY \end{pmatrix},
    \end{eqnarray*}   
where $\epsX, \epsY \in \bbR$. We fix $\epsilon_X = 0.5, \epsilon_Y=-0.3$ and $\beta = 0.2$. 

We also study the Conflict Network game, proposed in \citep{ewerhart:fp}, which does not satisfy Assumption \ref{ass::bimatrix-network}. In particular, each edge may be associated with a different bimatrix game. A network polymatrix game is a Conflict Network game if each bimatrix game $(A^{kl}, A^{lk})$ satisfies $(A^{kl})_{ij} = v^k (P^{kl})_{ij} - c^{kl}_i$ and $(A^{lk})_{ji} = v^j (P^{lk})_{ji} - c^{lk}_j$, 
where $v_k, v_l > 0$, $c_{kl} \in \bbR^{n_k}, c_{lk} \in \bbR^{n_l}$ and $(P^{kl})_{ij} + (P^{lk})_{ji} = 1$ for all $i \in
S_k$ and $j \in S_l$. To generate Conflict Network games, we independently sample each $v_k$ from the uniform distribution over $[0, 1]$ and generate $P^{kl}$ by randomly sampling its elements from the uniform distribution over $[-5, 5]$. We then construct $(A^{kl}, A^{lk})$ such that the Conflict Network condition is satisfied. 
\begin{figure*}[t!]
    \includegraphics[width=\textwidth]{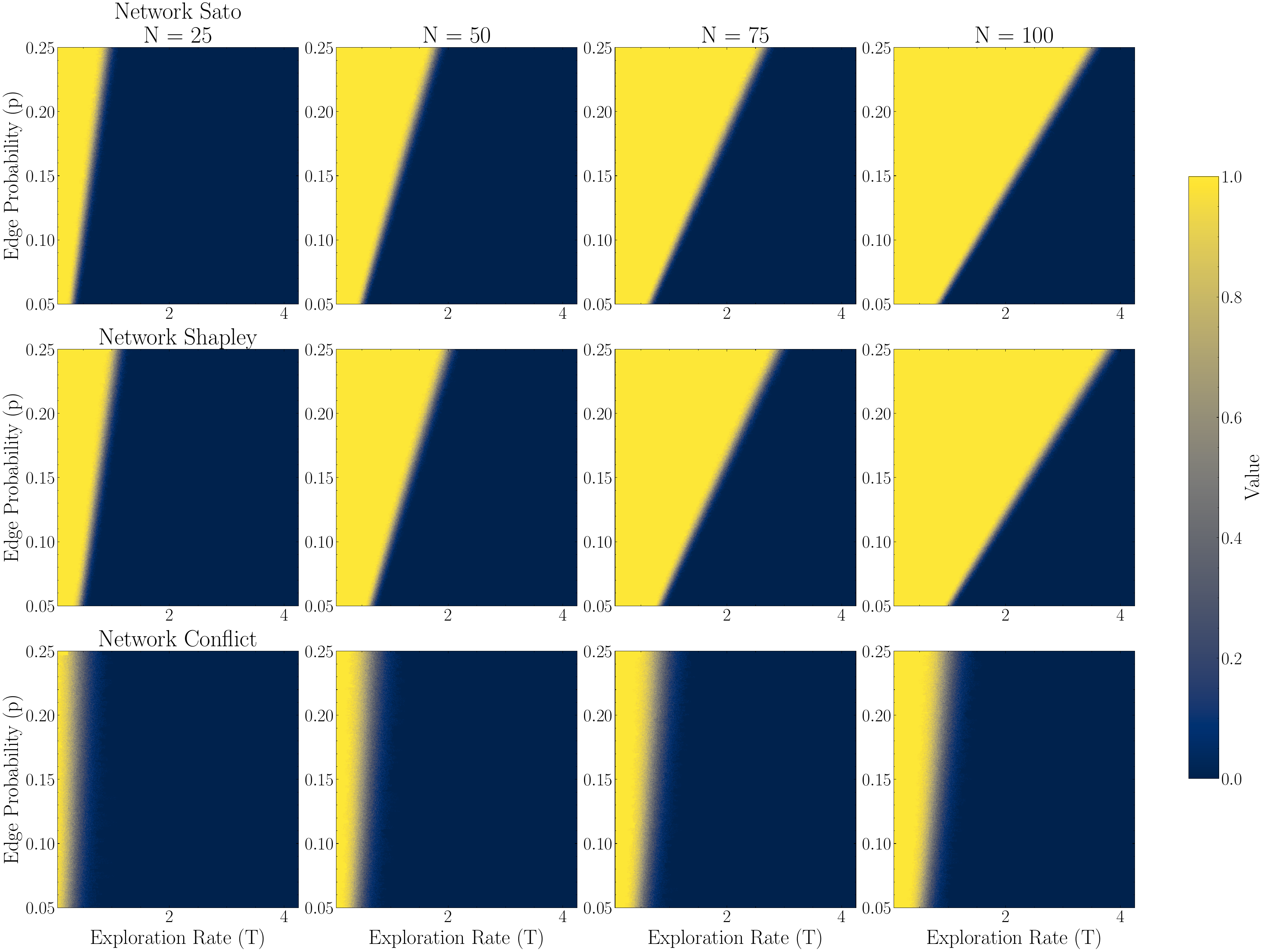}
        \caption{Proportion of (\ref{eqn::QLD}) simulations that diverge in network games with network drawn from the \er\ model using varying exploration rates $T$, edge probabilities $p$, and numbers of agents $N$. Each heatmap uses $(T,p)$ values on the grid $ [0.05,4.25] \times [0.05,0.25] $. Convergence at low exploration rates is more likely in sparser networks, i.e., low $p$. Further, the boundary between convergent and divergent behaviour shifts rapidly with the number of agents $N$ when $p$ is high, highlighting the need to control $p$ for large $N$. Additional heatmaps are available in Appendix \ref{appendix:extended_simulation_study}.}
        \label{fig::er-heatmaps}
\end{figure*}
\begin{figure*}[t!]
    \centering
    \begin{minipage}[b]{0.3\textwidth}
        \centering
        \includegraphics[width=\textwidth]{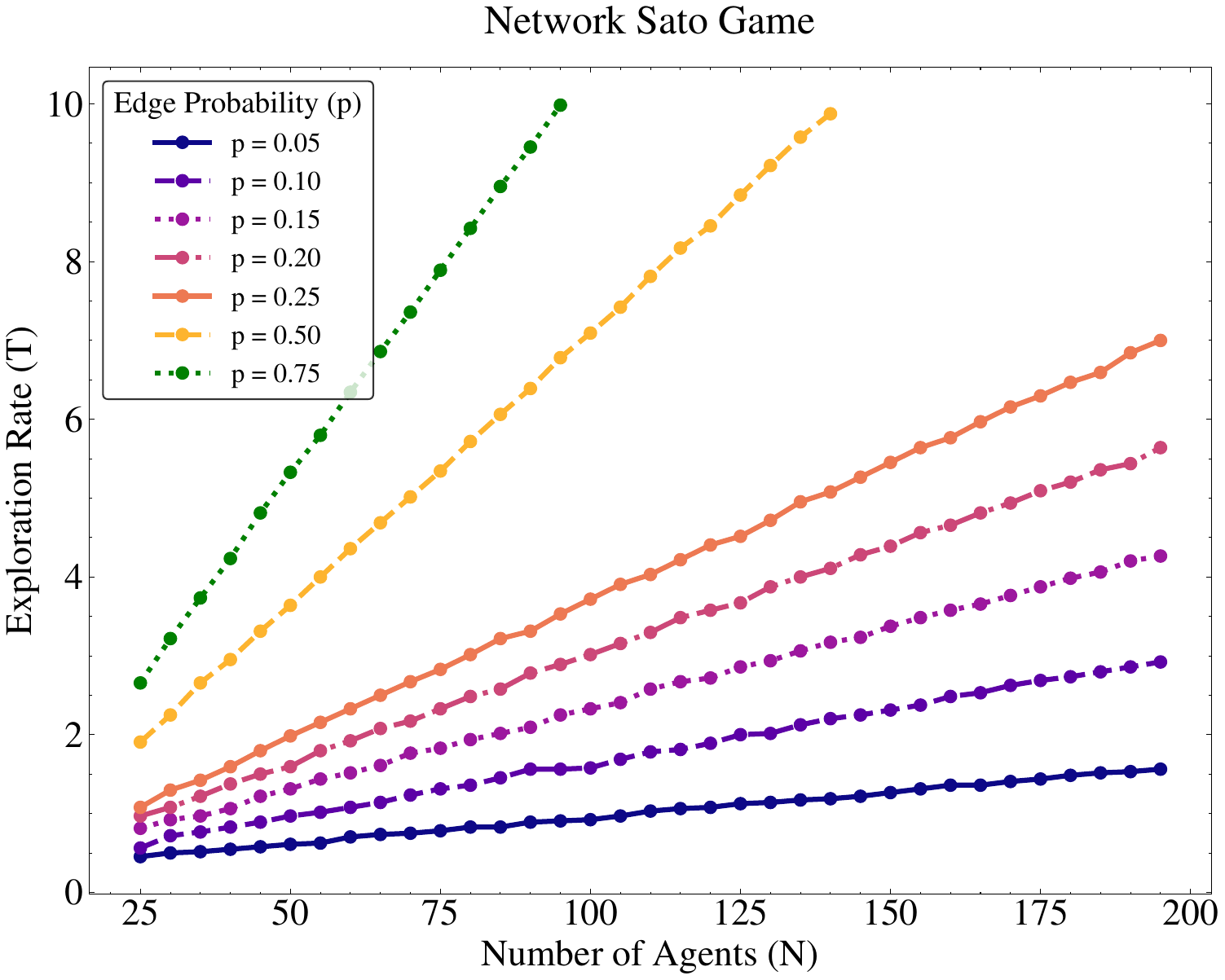}
    \end{minipage}
    \hfill
    \begin{minipage}[b]{0.3\textwidth}
        \centering
        \includegraphics[width=\textwidth]{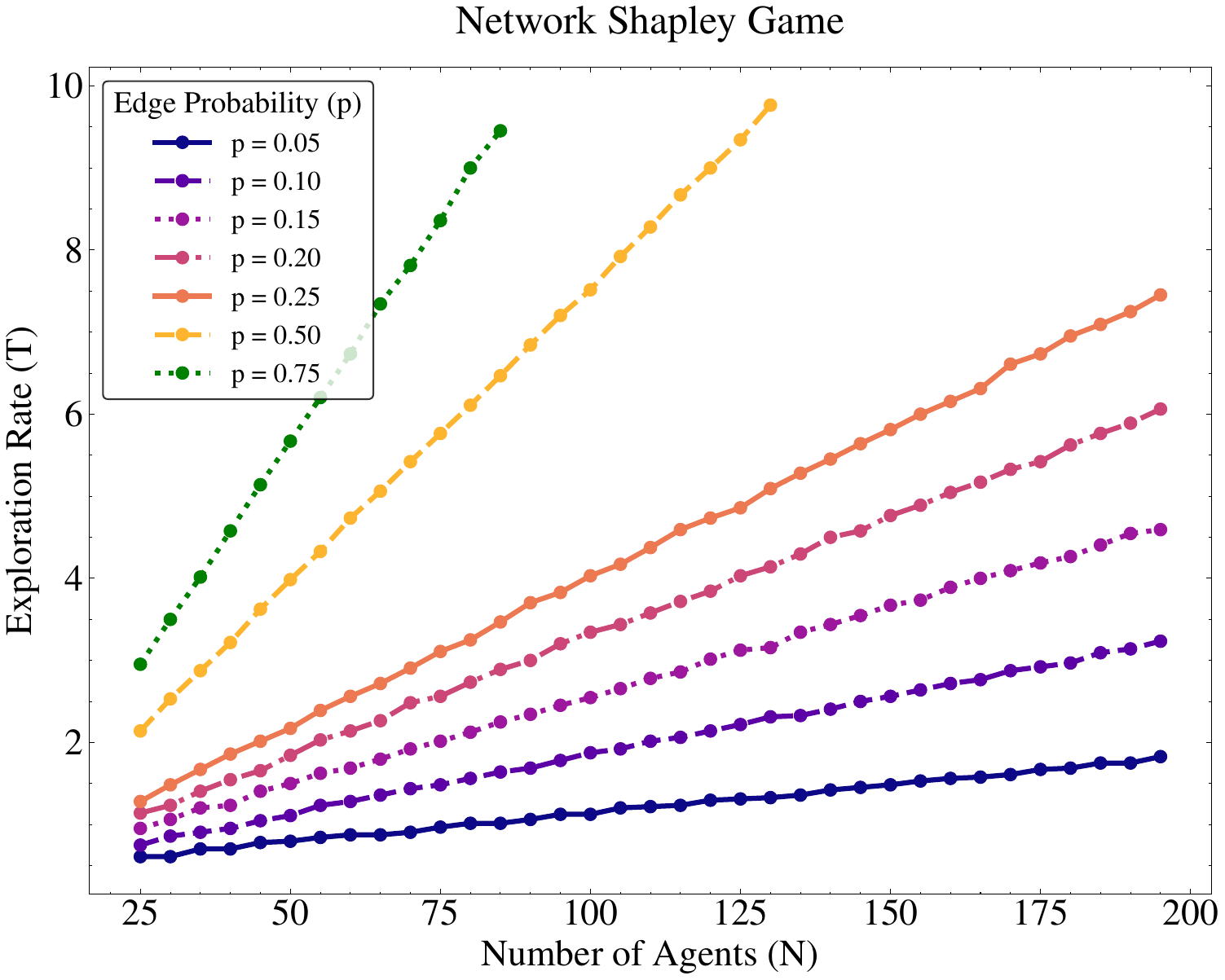}
    \end{minipage}
    \hfill
    \begin{minipage}[b]{0.3\textwidth}
        \centering
        \includegraphics[width=\textwidth]{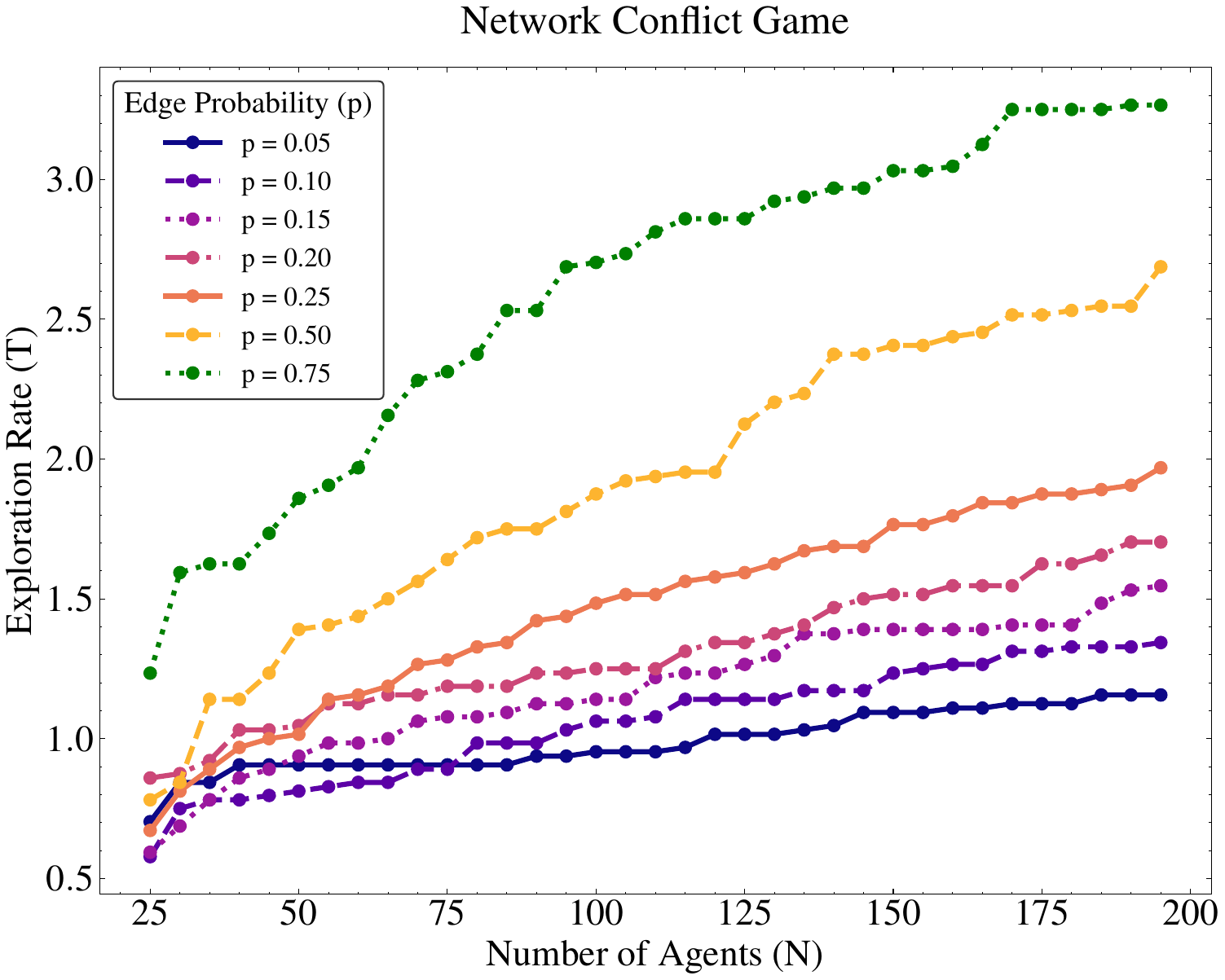}
    \end{minipage}
    
    \caption{Variation of the empirical convergence boundary in various network games as $N$ increases, for different values of $p$. The $y$-axis shows the smallest exploration rate for which \textbf{all} 50 simulations converged. Dense networks (high $p$) require substantially higher exploration rates to ensure convergence, while sparse networks (low $p$) maintain low convergence thresholds even for large $N$. Recall that lower exploration rates are preferable so that the QRE is closer to the Nash equilibrium.}
    \label{fig::area-plot}
\end{figure*}

\paragraph{\er\ Model.}In Figure \ref{fig::er-heatmaps}, we examine the convergence properties of Q-Learning on networks generated from the \er\ (ER) model. We vary the edge probability $p$, the exploration rate $T$, and the number of agents $N$. Figure~\ref{fig::area-plot} extends this analysis by plotting the convergence boundary as $N$ increases for different values of $p$. The key observation is that the rate of boundary growth increases with $p$, demonstrating that network density -- parameterised by the edge probability $p$ -- has a direct impact on learning dynamics. Our results reveal that dense networks (high expected node degree) require higher exploration rates to guarantee convergence, consistent with previous findings in \citep{sanders:chaos}. Recall from Section \ref{sec::solution-concepts}, however, that higher exploration rates yield Quantal Response Equilibria (QRE) that deviate further from the Nash Equilibrium. Crucially, we find that this effect is substantially reduced in sparse networks (small $p$).
\begin{figure*}[t!]
    \centering
    \includegraphics[width=\textwidth]{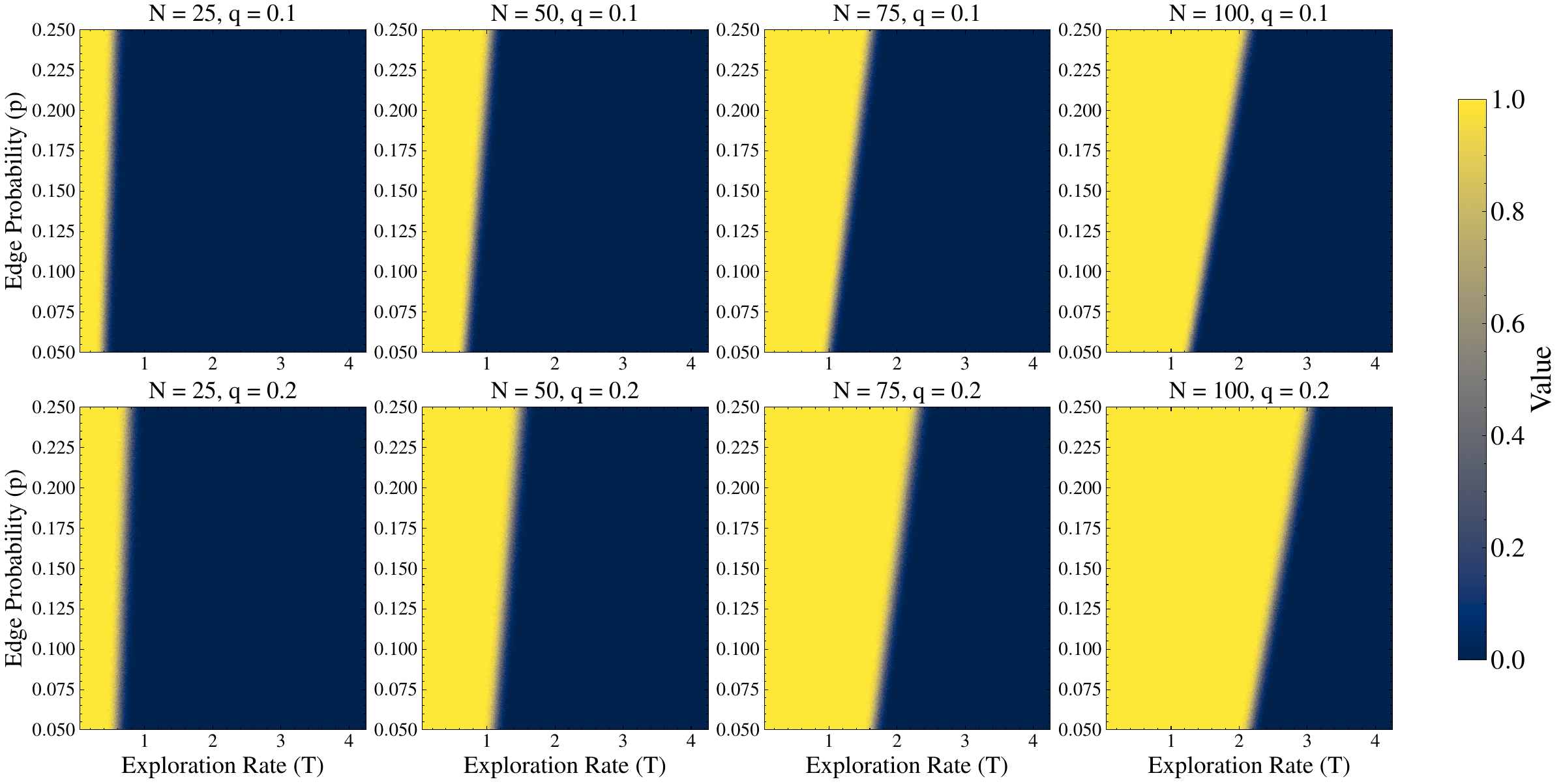}
    \caption{Proportion of (\ref{eqn::QLD}) simulations that diverge  in Network Sato games with network drawn from the Stochastic Block model, varying exploration rates $T$, intra-community edge probability $p$, inter-community edge probability $q$ and number of agents $N$. We use equidistant $(T,p)$ values in $[0.05,4.25] \times [0.05,0.25] $. The convergence boundary is jointly controlled by $p$ and $q$, offering greater flexibility to achieve convergence in structured networks. Extra results are provided in Appendix \ref{appendix:extended_simulation_study}.}
    \label{fig::sbm-heatmaps}
\end{figure*}
\begin{figure*}[t!]
    \centering
    \includegraphics[width=\textwidth]{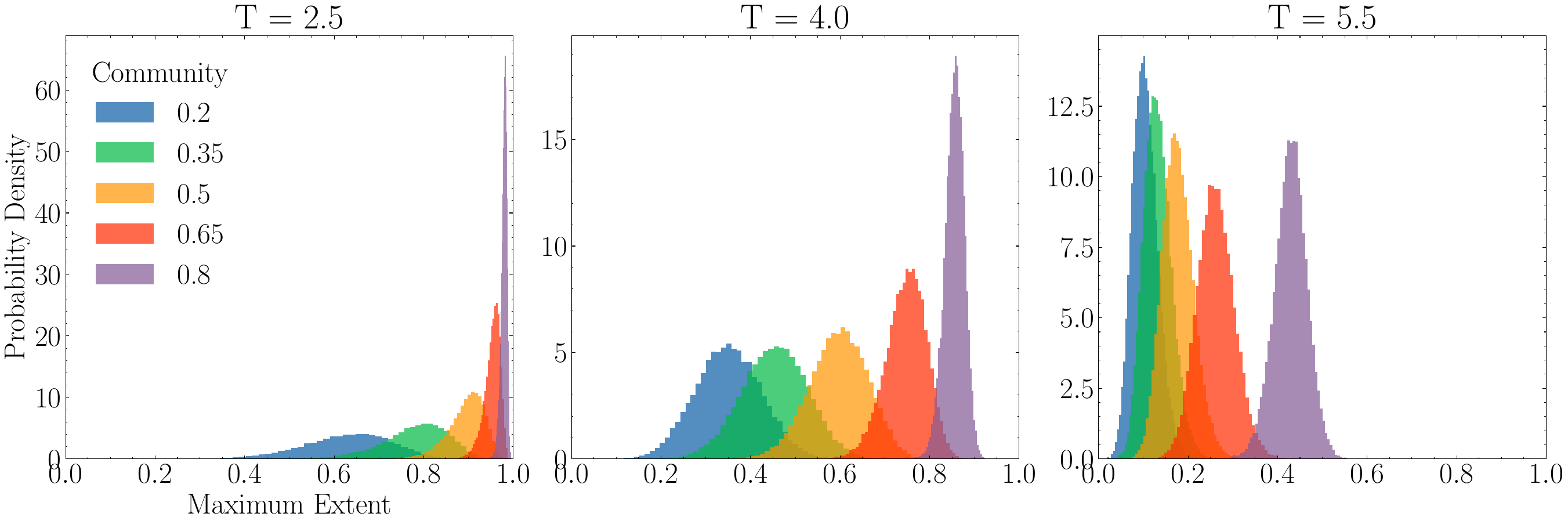}
    \caption{Probability density function of final strategy variation in Network Sato games on heterogeneous stochastic block networks with $N=200$ agents, showing the maximum strategy variation across agents during the final 300 iterations, computed from 1024 independent simulations. Networks contain five communities with varying intra-community connection probabilities $p$ (shown in legend) and fixed inter-community probability $q=0.1$. Communities with lower connectivity (blue, green) achieve convergence at lower exploration rates than densely connected communities (red, purple).}
    \label{fig::sbm-histograms}
\end{figure*}

The Network Sato game in Figure \ref{fig::area-plot} illustrates this phenomenon clearly. For $p = 0.05$, (\ref{eqn::QLD}) converges for $T < 2.0$ even with over $150$ agents. Conversely, for $p = 0.75$, convergence may fail with as few as 25 agents. These results also hold in the Conflict Network game, which extends beyond the setting of Assumption \ref{ass::bimatrix-network}. Although there is some stochasticity in the boundary due to randomised payoffs, the fundamental trend remains: convergence occurs at lower exploration rates when the expected node degree is lower. These findings underscore the importance for practitioners of carefully designing network structures to ensure multi-agent learning produces stable, near-optimal equilibria.
\paragraph{Stochastic Block Model.} We next evaluate Q-Learning on networks generated from the Stochastic Block (SB) model. Figure~\ref{fig::sbm-heatmaps} shows the proportion of diverged experiments in the Network Sato game, with results for the Network Shapley and Conflict games provided in Appendix~\ref{appendix:extended_simulation_study}.
Because the SB model has many parameters, we fix the number of agents in
each community to $N_c = 5$ and suppose that $p_c = p$ for all communities $c$. We then assess
the convergence of (\ref{eqn::QLD}) for different values of $p$ and $T$, while varying $q$ and
$N$. We find that, although each parameter influences last-iterate behaviour, convergence occurs less frequently in denser graphs, i.e., as $p$ and $q$ increase. Finally, we examine the case in which $p_c$ varies across communities in Figure~\ref{fig::sbm-histograms}. We simulate $1024$ independent runs of \ref{eqn::QLD} dynamics and compute the maximum difference between strategy components over the final iterations. The results show that communities with lower $p_c$ tend to exhibit smaller maximum differences. This suggests that allowing agents in different communities to use different exploration rates may yield convergence with smaller rates than required by Theorem \ref{thm::sbm-convergence}.

\section{Conclusion}
This paper examined Q-Learning dynamics in network polymatrix games through the lens of random networks. In doing so, we showed an explicit relationship between the last-iterate
convergence of Q-Learning and the expected degree of nodes in the underlying network. We provided theoretical lower
bounds on the exploration rates required for convergence in network polymatrix games where the
network is drawn from the \er\ or Stochastic Block models and showed that the convergence may be
guaranteed in games with arbitrarily many agents, so long as interactions in the network are controlled.
\paragraph{Future Work.} In our experiments, we found that the implications of our theorems hold in settings beyond the assumptions under which they were derived. This suggests several directions for future work. One is to examine if our results generalise to broader classes of network games, e.g., games with
continuous action sets. In addition, whilst our work
considers the repeated play of a matrix game, extending these results to the \emph{Markov Game} framework -- which introduces a state variable -- would enable theoretical guarantees to be placed in real-world multi-agent systems. Of further interest is the extension to other random network models, e.g., the Barabási-Albert model, used to model systems such as the Internet that exhibit preferential attachment. Notably, a direct extension to our current proof technique may be applied to random networks with dependent entries (see end of Appendix~\ref{app::proof2}). Analysing these models could reveal deeper connections between network structure and the convergence of learning algorithms.

\newpage

\bibliographystyle{plainnat}
\bibliography{references}

\newpage
\appendix
\section*{Appendix}
The appendix is organized as follows:
\begin{itemize}
    \item Appendix \ref{appendix:A} presents the necessary preliminaries for our proofs, including definitions and properties of \emph{monotone games} and certain properties of matrix norms.
    \item Appendix \ref{app::proof1} contains the proof of Lemma \ref{lem::qld-spectral-radius}.
    \item Appendix \ref{app::proof2} contains the proofs of Lemmas \ref{lem::er-spectral-radius} and \ref{lem::sbm-spectral-radius} and establishes asymptotic upper bounds on the spectral radius of networks drawn from the \er\ and Stochastic Block models.
    \item Appendix \ref{app::undirected} describes how payoffs are assigned to edges in the network polymatrix game under Assumption \ref{ass::bimatrix-network}. We clarify that the underlying network is undirected. 
    \item Appendix \ref{appedinx_further_sim_results} contains further simulation studies and gives details on the computational schemes used to produce the visualizations from the main body.
    \item Appendix \ref{appendix:Q_learning_dynamics} gives details on the continuous-time Q-Learning Dynamic (\ref{eqn::QLD}) and its relation to the underlying \emph{independent Q-Learning} algorithm discussed in Section \ref{sec::q-learning}.
\end{itemize}
\textbf{Disclaimer: } LLMs were used in this work for grammar checks and generally polishing the text and code.
\newpage
\section{Variational Inequalities and Monotone Games}\label{appendix:A}
The main idea of our convergence proof is to show that, under the setting of Lemma \ref{lem::ql-conv}, the corresponding network game is a \emph{strictly monotone} game \cite{melo:network,parise:network,mertikopoulos:finite,sorin:composite}. In such games, it is known that the equilibrium solution is unique \cite{melo:qre,facchinei:VI}. It is further known, from \cite{hussain:aamas}, that (\ref{eqn::QLD}) converges asymptotically in monotone games. We leverage this result to prove Lemma \ref{lem::qld-spectral-radius}.

We begin by framing game theoretic concepts in the language of variational inequalities.

\begin{definition}[Variational Inequality] 
     Consider a set $\clX \subset \bbR^d$ and a map $F \, : \clX \rightarrow \bbR^d$. The Variational
     Inequality (VI) problem $VI(\clX, F)$ is given as
        \begin{equation}\label{eqn::VIdef}
            \langle \bfx - \bfx^*, F(\bfx^*) \rangle \geq 0, \hspace{0.5cm} \text{ for all } \bfx \in \clX.
        \end{equation}
    We say that $\bfx^* \in \clX$ belongs to the set of solutions to a variational inequality problem
    $VI(\clX, F)$ if it satisfies (\ref{eqn::VIdef}).
\end{definition}

We now wish to reformulate the problem of finding Quantal Response Equilibria, or Nash Equilibria, as a problem of solving a variational inequality of a particular form. In such a case, the set $\clX$ is identified with the joint simplex $\Delta$. The map $F$ is identified with the \emph{pseudo-gradient} map of the game.

\begin{definition}[Pseudo-Gradient Map]
    The $\clG$ be a network polymatrix game with payoff functions $(u_k)_{k \in \clN}$. Then, the pseudo-gradient map of $\clG$ is $F : \bfx \mapsto (-D_{\bfx_k}u_k(\bfx_k, \bfx_{-k}))_{k \in \clN}$.
\end{definition}

This reformulation has been used, for example, by \cite{melo:qre} to show that a QRE of a game can be found by solving a variational inequality of a particular form. We reformulate their theorem for the particular case of network polymatrix games. 
\begin{lemma}[\cite{melo:qre}]
    Consider a game $\clG = (\clN, \clE, (\clA_k)_{k \in \clN}, (A^{kl}, A^{lk})_{(k, l) \in \clE})$ and for
    any $T_1, \ldots, T_N > 0$, let the \emph{regularised} game $\clG^H$ be the network game in which the payoff $u_k^H$ to each agent $k$ is given by
    \begin{equation} 
        u_k^H(\bfx_k, \bfx_{-k}) = \sum_{l:(k, l) \in \clE} \bfx_k^\top A^{kl} \bfx_l - T_k\langle \bfx_k, \ln \bfx_k \rangle.
    \end{equation}
    Now, let $F^H$ be the pseudo-gradient map of $\clG^H$. Then $\bfx^* \in \Delta$ is a QRE of $\clG$ if and only if $\bfx^*$ is a solution to $VI(\Delta, F^H)$. .
\end{lemma}
In fact, this same lemma can be used to show that $\bfx^*$ is a QRE of $\clG$ if and only if it is a Nash Equilibrium of $\clG^H$. This concept of regularised games has also been used to show connections between the replicator dynamics \cite{smith:replicator} and Q-Learning dynamics in \cite{piliouras:potential}, and to design algorithms for Nash Equilibrium seeking \cite{gemp:sample}. We require one final component from the study of variational inequalities which is often used to study uniqueness of equilibrium solutions: monotonicity \cite{parise:network,melo:network}.
\begin{definition}[Monotone Games]
    Let $\clG$ be a network polymatrix game with pseudo-gradient map $F$. $\clG$ is said to be 
    \begin{enumerate}
        \item \emph{Monotone} if, for all $\bfx, \bfy \in \Delta$,
        \begin{equation*}
            \langle F(\bfx) - F(\bfy), \bfx - \bfy \rangle \geq 0.
        \end{equation*}
        \item \emph{Strictly Monotone} if, for all $\bfx, \bfy \in \Delta$,
        \begin{equation*}
            \langle F(\bfx) - F(\bfy), \bfx - \bfy \rangle > 0.
        \end{equation*}
        \item \emph{Strongly Monotone} with constant $\alpha > 0$ if, for all $\bfx, \bfy \in \Delta$,
        \begin{equation*}
            \langle F(\bfx) - F(\bfy), \bfx - \bfy \rangle \geq \alpha ||\bfx - \bfy||^2_2.
        \end{equation*}
    \end{enumerate}
\end{definition}
By proving monotonicity properties of a game, we can leverage the following results from literature.
\begin{lemma}[\cite{melo:qre}] \label{lem::unique-qre} Consider a game $\clG$ and for any $T_1, \ldots, T_N > 0$, let $F$ be the pseudo-gradient map of $\clG^H$. $\clG$ has a unique QRE $\bfx^* \in \Delta$ if $\clG^H$ is strongly monotone with any $\alpha>0$.
\end{lemma}
\begin{lemma}[\cite{hussain:aamas}] \label{lem::ql-conv} If the game $G$ is \emph{monotone}, then the Q-Learning Dynamics (\ref{eqn::QLD}) converges to a unique QRE $\bfx^* \in \Delta$ with any positive exploration rates $T_1, \ldots, T_N > 0$.
\end{lemma}
Finally, note that a  map $g : \Delta \rightarrow \bbR$ is \emph{strongly convex}
with constant $\alpha$ if, for all $\bfx, \bfy \in \Delta$
\begin{equation*}
    g(\bfy) \geq g(\bfx) + Dg(\bfx)^\top (\bfy - \bfx) + \frac{\alpha}{2}\lVert\bfx - \bfy \rVert^2_2.
\end{equation*}
If the map $g(\bfx)$ is twice-differentiable, then it is strongly convex if its Hessian $D^2_{\bfx} g(\bfx)$ is strongly
positive definite with constant $\alpha$. Thus, all eigenvalues of $D^2_{\bfx} g(\bfx)$ are larger than
$\alpha$. It also holds that, if $D^2_{\bfx} g(\bfx)$ is strongly positive definite, the gradient $D_{\bfx} g(\bfx)$ is strongly monotone. To apply this in our setting, we use the following result.
\begin{proposition}[\cite{melo:qre}] \label{prop::strong-conv} The function $g(\bfx_k) = T_k \langle
        \bfx_k, \ln \bfx_k \rangle$ is strongly convex with constant $T_k$.
\end{proposition}
\subsection{Matrix Norms}
We close with a few properties of matrices that are useful towards our decomposition of the payoff matrices and the graph adjacency matrix.
\begin{proposition}
    For any matrix $A$, $\lVert A \rVert_2 = \sqrt{\lambda_{\max}(A^
    \top A)}$, where $\lambda_{\max}(\cdot)$ denotes the largest eigenvalue of a matrix. If, in addition, the matrix is symmetric, $\lVert A \rVert_2 = \rho(A)$, where $\rho(A)$ is the spectral radius of $A$.
\end{proposition}

\begin{proposition}[Weyl's Inequality]
    Let $J = D + N$ where $D$ and $N$ are symmetric matrices. Then it holds that
    \begin{equation*}
        \lambda_{\min}(J) \geq \lambda_{\min}(D) + \lambda_{\min}(N),
    \end{equation*}
    where $\lambda_{\min}(\cdot)$ denotes the smallest eigenvalue of a matrix.
\end{proposition}
\begin{proposition} \label{prop::kron-two-norm}
    Let $G, A$ be matrices and let $\otimes$ denote the Kronecker product. Then
    \begin{equation}
        \lVert G \otimes A \rVert_2 = \lVert G \rVert_2 \lVert A \rVert_2
    \end{equation}
\end{proposition}
\begin{proposition}
    Let $A$ be a symmetric matrix. Then \[
    |\lambda_{\min}(A)| \leq \rho(A)
    \]
\end{proposition}

\begin{lemma}
    Let $A, B \in \mathcal{M}_{m,n}\left(\mathbb{R}\right)$ such that  $ 0 \le A_{ij} \le B_{ij}$ for any $1 \le i \le m, \, 1 \le j \le n$. Then 
\begin{equation*}
\lVert A \rVert_2 \le \lVert B \rVert_2.
\end{equation*} \label{lem::operator_norm_ineq_for_nonnegative_matrices}
\end{lemma}
\begin{proof}
Let $\mathbb{R}^n_{+}$ be the set of $n$-dimensional nonnegative vectors, i.e.~with nonnegative entries. As $A$ and $B$ have nonnegative entries, we deduce that
\begin{align*}
   \lVert A \rVert_2 =  \sup_{\bfx \in \mathbb{R}^n} \lVert A \bfx \rVert_2 = \sup_{\bfx \in \mathbb{R}^n_{+}} \lVert A \bfx \rVert_2, \\
    \lVert B \rVert_2 =  \sup_{\bfx \in \mathbb{R}^n} \lVert B \bfx \rVert_2 = \sup_{\bfx \in \mathbb{R}^n_{+}} \lVert B \bfx \rVert_2.
\end{align*}
Further, for any $\bfx \in \mathbb{R}^n_{+}$, it holds that
\begin{equation*}
     \lVert B \bfx \rVert_2  = \lVert (B-A) \bfx + A \bfx \rVert_2 \ge \lVert A \bfx \rVert_2, 
\end{equation*}
as $(B-A)\bfx$ is nonegative. This increases the norm of any nonnegative vector $A\bfx$. Taking supremum in the above equation, we obtain that
\begin{equation*}
    \sup_{\bfx \in \mathbb{R}^n} \lVert A \bfx \rVert_2 =  \sup_{\bfx \in \mathbb{R}_{+}^n} \lVert A \bfx \rVert_2  \le \sup_{\bfx \in \mathbb{R}_{+}^n} \lVert B \bfx \rVert_2 =  \sup_{\bfx \in \mathbb{R}^n} \lVert B \bfx \rVert_2 .
\end{equation*}
\end{proof}
\newpage
\section{Proof of Lemma \ref{lem::qld-spectral-radius}}
\label{app::proof1}

The proof takes the following steps. We first decompose the derivative of the pseudo-gradient, which we dub the \emph{pseudo-jacobian}, of the regularised game $\clG^H$ into a term involving payoffs and a term involving exploration rates. In doing so, we can determine how the exploration rates should be balanced so that the pseudo-jacobian is positive definite, which yields monotonicity of the game. We then further decompose the payoff term into terms involving the payoff matrices and the network adjacency matrix. This exposes the connection between each of the three quantities: exploration rate, payoff matrices and network connectivity.

\begin{proof}[Proof of Lemma \ref{lem::qld-spectral-radius}]
    Let $F$ be the pseudo-gradient of the regularised game $\clG^H$. We define the pseudo-jacobian as the derivative of $F$, given by
    \[
    [J(\bfx)]_{k, l} = D_{\bfx_l} F_k(\bfx)
    \]
    It holds that if $\frac{J(\bfx) + J^\top(\bfx)}{2}$ is positive definite for all $\bfx \in
    \Delta$ then $F(\bfx)$ is monotone. We decompose $J$ as 
    \begin{equation*}
        J(\bfx) = D(\bfx) + N(\bfx),
    \end{equation*}
    where $D(\bfx)$ is a block diagonal matrix with $-D^2_{\bfx_k \bfx_k} u_k^H(\bfx_k, \bfx_{-k})$ along the
    diagonal. $N(\bfx)$ is an off-diagonal block matrix with
    \begin{equation*}
        [N(\bfx)]_{k, l} = \begin{cases}
            - D_{\bfx_k, \bfx_l} u_k^H(\bfx_k, \bfx_{-k}) &\text{ if } (k, l) \in \clE \\
            \mathbf{0} &\text{ otherwise}
        \end{cases}.
    \end{equation*}
    Now notice that $-u_k^H(\bfx_k, \bfx_{-k}) = T_k \langle \bfx_k, \ln \bfx_{-k} \rangle - \sum_{l:(k, l) \in \clE} \bfx_k^{\top} A^{kl} \bfx_l$. Therefore, $D(x)$ is simply the Hessian of the entropy regularisation term $T_k \langle \bfx_k, \ln \bfx_k \rangle$. From Proposition \ref{prop::strong-conv}, it holds then that $D(\bfx)$ is strongly positive definite with constant $T = \min_k T_k$. Let $\bar{J}(\bfx)$ be defined as
    \begin{equation*}
        \Bar{J}(\bfx) = D(\bfx) + \frac{N(\bfx) + N^\top(\bfx)}{2}.
    \end{equation*}
    In words, $\Bar{J}(\bfx)$ is the symmetric component of $J(\bfx)$. We may now use Weyl's inequality to write
    \begin{align*}
        \lambda_{\min}(\Bar{J}(\bfx)) &\geq \lambda_{\min}(D(\bfx)) + \lambda_{\min}\left(
        \frac{N(\bfx) + N(\bfx)^\top}{2}
        \right) \\
        &\geq T - \rho \left(\frac{N(\bfx) + N(\bfx)^\top}{2} \right) \\
        &= T - \frac{1}{2} \lVert N(\bfx) + N(\bfx)^\top \rVert_2 \\
    \end{align*}

    To determine $\lVert N(\bfx) + N(\bfx)^\top \rVert_2$, we first notice that, in network polymatrix games each block of $N(\bfx)$ is given by
    \[
    [N(\bfx)]_{k, l} = \begin{cases}
        -A^{kl} &\text{ if } (k, l) \in \clE \\
        \mathbf{0} \text{ otherwise.}
    \end{cases}
    \]
    whilst
    \[
    [N(\bfx)^\top]_{k, l} = \begin{cases}
        -(A^{lk})^\top &\text{ if } (k, l) \in \clE \\
        \mathbf{0} \text{ otherwise.}
    \end{cases}
    \]
    To write this in the form of a kronecker product, we leverage Assumption \ref{ass::bimatrix-network}, namely that each edge corresponds to the same bimatrix game with payoff matrices $(A, B)$. 
    We decompose each edge into a half-edge along which $A$ is played, and a half-edge along which $B$ is played. 
    In doing so, we may decompose the adjacency matrix $G = G_{k \rightarrow l} + G_{l \rightarrow k}$. The non-zero elements of $G_{k \rightarrow l}$ correspond to half edges along which $A$ is played, $G_{l \rightarrow k}$ denote the half-edges along which $B$ is played. With this definition in place, we may write
    \[
    N(\bfx) + N(\bfx)^\top = -(A + B^\top) \otimes G_{k \rightarrow l} - (A^\top + B) \otimes G_{l \rightarrow k}
    \]
    Then, by Proposition \ref{prop::kron-two-norm}
    \begin{align*}
    \frac{1}{2} \lVert N(\bfx) + N(\bfx)^\top \rVert_2 &= \frac{1}{2} \lVert (A + B^\top) \otimes G_{k \rightarrow l} + (A^\top + B) \otimes G_{l \rightarrow k} \rVert_2 \\
    &\leq \frac{1}{2} \lVert (A + B^\top) \otimes G_{k \rightarrow l} \rVert_2 + \frac{1}{2} \lVert (A^\top + B) \otimes G_{l \rightarrow k} \rVert_2 \\
    &=\frac{1}{2} \lVert A + B^\top\rVert_2 (\lVert G_{k \rightarrow l}\rVert_2 + \lVert G_{l \rightarrow k}\rVert_2) \\
    &\leq \lVert A + B^\top\rVert_2 \lVert G \rVert_2 \\
    &= \delta_I \rho(G)
    \end{align*}
    where in the final inequality, we use Lemma \ref{lem::operator_norm_ineq_for_nonnegative_matrices}. From this we may establish that
    \begin{align*}
        \lambda_{\min}(\Bar{J}(x)) & \geq T - \delta_I \rho(G)
    \end{align*}

    Therefore, if $T > \delta_I \rho_G$, the pseudo-jacobian is strongly positive definite with constant $T - \delta_I \rho(G)$, from which it is established that $\clG^H$ is strongly monotone with the same constant. From Lemma \ref{lem::unique-qre} \cite{melo:qre} the QRE is unique, and from Lemma \ref{lem::ql-conv}, it is asymptotically stable under (\ref{eqn::QLD}).
\end{proof}

\newpage
\section{Proofs of Lemmas \ref{lem::er-spectral-radius} and \ref{lem::sbm-spectral-radius}}
\label{app::proof2}
We now focus on establishing upper bounds on the spectral radius when the network is drawn from the \er\ or Stochastic Block models. The proof idea is to decompose $G$ into $\bbE[G]$ and $\Tilde{G} = 
 G - \bbE[G]$. Then, $\rho(\bbE[G])$ is deterministic and can be computed in closed form. For $\rho(\Tilde{G})$, we require Lemma \ref{lemma_we_need}, which relies on Bernstein's matrix inequality \citep{tropp}. 
 
 We structure this appendix as follows: first list the required results (Lemmas \ref{lemma:bernstein} and \ref{lemma_we_need} and Corollary \ref{corollary_bern}), then give proofs for Lemmas \ref{lem::er-spectral-radius} and \ref{lem::sbm-spectral-radius} and then the proof for the key Lemma \ref{lemma_we_need}. Finally, we comment on the possibility to extend these results to the dependent entry case.
 
\begin{lemma}[Bernstein] Let $Y^{(1)},\ldots,Y^{(K)}$ be independent and symmetric $N \times N$ random matrices with zero-mean entries, i.e., $\bbE [Y^{(k)}_{ij}]$ = 0 almost surely for any $1 \le i,j \le N$ and $1 \le k \le K$. Define 
\begin{align*}
    \tilde{G} =& Y^{(1)} + \cdots + Y^{(K)},\\
    v^2 =& \rho\left(\bbE[\tilde{G}^2]\right) =  \rho \left(\bbE[(Y^{(1)})^2+\cdots + (Y^{(K)})^2]\right). 
\end{align*}
If $\rho(Y^{(k)}) \le L$ for all $ 1 \le k \le K$, then the following holds for any $t > 0$
\begin{equation}
    P(\rho( \tilde{G}) > t) \le 2N \exp{\left(\frac{-t^2/2}{v^2 + Lt/3}\right)}.\label{eq:bernstein}
\end{equation}
\label{lemma:bernstein}
\end{lemma}
\begin{corollary}We will use a different form of the above inequality, which follows by setting the RHS of \eqref{eq:bernstein} to $\epsilon$ and solving for $t$ in terms of $\epsilon$. For any $\epsilon > 0,$ it holds that
\begin{equation*}
    \rho( \tilde{G}) \le \sqrt{2 v^2 \log{\frac{2N}{\epsilon}}} + \frac{2L}{3}\log{\frac{2N}{\epsilon}} \text{ with probability at least } 1-\epsilon.
\end{equation*} 
\label{corollary_bern}
\end{corollary}
\begin{lemma}\label{lemma_we_need}
    Let $G$ be the adjacency matrix of a random graph with $0$s on the main diagonal and with independent entries $[G]_{ij}\sim \texttt{Bernoulli}(p_{ij})$, apart from $[G]_{ij}$ and $[G]_{ji}$, which are equal. Let $\tilde{G} = G - E[G]$ with entries $[\tilde{G}]_{ij} = [G]_{ij} - p_{ij}$. Then $\tilde{G}$ satisfies Corollary \ref{corollary_bern} with 
    \begin{align*}
        v^2 =& \max_{1 \le i \le N} \sum_{j \neq i} p_{ij} (1-p_{ij}) \\ 
        L   =& \max_{1 \le i < j \le N} \{p_{ij},1-p_{ij}\} \le 1.
    \end{align*}
For simplicity, as $L$ only appears in the non-dominant term, we use $L=1$. In the ER case, this amounts to
\begin{equation*}
    \rho(\tilde{G}) \le \sqrt{2(N-1)\, p \, (1-p) \log{\frac{2N}{\epsilon}}  } + \frac{2}{3} \log{\frac{2N}{\epsilon}} \text{ with probability at least } 1 - \epsilon.
\end{equation*}
In the SBM case, this amounts to
\begin{equation*}
    \rho(\tilde{G}  ) \le \sqrt{2 \left( (N-N/C)  q (1-q) + (N/C-1) \sigma_{\text{p,max}}^2 \right) \log{\frac{2N}{\epsilon}}} + \frac{2}{3} \log{\frac{2N}{\epsilon}}
\end{equation*}
with probability at least $1 - \epsilon$ and where $\sigma_{\text{p,max}} = \max_{1 \le l \le C} \sqrt{p_l(1-p_l)}$.
\end{lemma}
Proofs for Corollary \ref{corollary_bern} and Lemma \ref{lemma_we_need} can be found at the end of the section.
 \begin{proof}[Proof of Lemma \ref{lem::er-spectral-radius}]
For the deterministic part, we have that  
\begin{equation}
\bbE[G] =
\begin{pmatrix}
0 & p & p & \cdots & p \\
p & 0 & p & \cdots & p \\
p & p & 0 & \cdots & p \\
\vdots & \vdots & \vdots & \ddots & \vdots \\
p & p & p & \cdots & 0
\end{pmatrix} = p J_N - p I_N, \label{eq:det_part}
\end{equation}
with eigenvalues $\lambda_1 = (N-1)p$ and $\lambda_2 = \ldots = \lambda_N = -p$. Hence $\rho(\bbE[G]) = (N-1)p$. Further, by applying Lemma~\ref{lemma_we_need}, we obtain that with probability at least $1-\epsilon$
    \begin{equation*}
        \rho(G) \leq \rho(\bbE[G]) + \rho\left(\tilde{G}\right) \le (N-1)p + \sqrt{2(N-1)\, p \, (1-p) \log{\frac{2N}{\epsilon}}  } + \frac{2}{3} \log{\frac{2N}{\epsilon}}.
    \end{equation*}
\end{proof}
\begin{proof}[Proof of Lemma \ref{lem::sbm-spectral-radius}]
Let $I_N$ be the identity matrix and $J_N$ be the matrix whose elements are all ones. For the deterministic part, we have that  
\[
\bbE[G] = 
\begin{pmatrix}
P_1 & Q & Q & \cdots & Q \\
Q & P_2 & Q & \cdots & Q \\
Q & Q & P_3 & \cdots & Q \\
\vdots & \vdots & \vdots & \ddots & \vdots \\
Q & Q & Q & \cdots & P_C
\end{pmatrix},
\]
where $Q = q J_N$ and $P_l = pJ_m - p I_m$ for $1 \le l \le C$, as in \eqref{eq:det_part}. Thus
\begin{align*}
    \rho\left(\bbE[G]\right) \leq \lVert\bbE[G]\rVert_1 = & (N-N/C)q+\max_{1 \leq l \leq C}{\{(N/C-1)p_l\}} = (N-N/C)q + (N/C-1)p_{\text{max}} \\
     = &Nq + N(p_{\text{max}}-q)/C-p_{\text{max}},
     \end{align*}
where $p_{\text{max}} = \max_{1\leq l \leq C} p_l$. Let $ \sigma_{\text{p,max}}= \max{\{\sqrt{p_1(1-p_1)}, \ldots,  \allowbreak \sqrt{p_C(1-p_C)}\}}$. By applying Lemma \ref{lemma_we_need} for $\tilde{G}$ and combining the two upper bounds, we obtain that 
\begin{align*}
   \rho(G) \leq \rho(\bbE[G]) + \rho\left(\tilde{G}\right) & \leq  (N-N/C)q + (N/C-1)p_{\text{max}} + \\
   &  \sqrt{2 \left( (N-N/C)  q (1-q) + (N/C-1) \sigma_{\text{p,max}}^2 \right) \log{\frac{2N}{\epsilon}}} + \frac{2}{3} \log{\frac{2N}{\epsilon}}.
\end{align*}
\end{proof}
\begin{proof}[Proof of Corollary \ref{corollary_bern}] 
Let $\epsilon = 2N \exp{\left(\frac{-t^2/2}{v^2 + Lt/3}\right)}$. Rearranging, we get a unique solution 
\begin{equation*}
    t = \frac{L}{3}\log\frac{2N}{\epsilon} + \sqrt{\frac{L^2}{9}\left(\log\frac{2N}{\epsilon}\right)^2 + 2v^2\log\frac{2N}{\epsilon}}
\end{equation*}
We thus have that for any $\epsilon >0$
\begin{equation*}
    P \left(\rho(A) > \frac{L}{3}\log\frac{2N}{\epsilon} + \sqrt{\frac{L^2}{9}\left(\log\frac{2N}{\epsilon}\right)^2 + 2v^2\log\frac{2N}{\epsilon}} \right) \le 1 -\epsilon.
\end{equation*}
For simplicity, we use the slightly looser bound
\begin{equation*}
    P\left (\rho(A) > \frac{2L}{3}\log{\frac{2N}{\epsilon}} + \sqrt{2 v^2 \log{\frac{2N}{\epsilon}}} \right)  \le 1- \epsilon,
\end{equation*}
obtained by applying $\sqrt{a^2 + b} < a  + \sqrt{b}$ for $a,b \ge 0,$ to the previous inequality.
\end{proof}
\begin{proof}[Proof of Lemma \ref{lemma_we_need}]
    Notation-wise, let $e_i$ be the canonical basis vectors for $1 \le i \le N$. For each $1 \le i < j \le N$, define $Y^{ij} = [\tilde{G}]_{ij} (e_i e_j^T + e_j e_i^T)$, which keeps the $\text{ij}^{th}$ and $\text{ji}^{th}$ elements of $\tilde{G}$ and sets everything else to $0$. Then 
    \begin{equation*}
        \tilde{G} = \sum_{1 \le i < j \le N} Y^{ij}.
    \end{equation*}
    By definition, the matrices $Y^{ij}$ are independent, symmetric random matrices and $\bbE [Y^{ij}] =0$ for $1 \le i < j \le N$. We now determine $L$ and $v^2$. 

    For $L$, note that $Y_{ij}$ has two non-zero eigenvalues, i.e., $[\tilde{G}]_{ij} = [G]_{ij} - p_{ij}$ and $-[\tilde{G}]_{ij} = - [G]_{ij} + p_{ij}$ with corresponding eigenvectors $e_i+e_j$ and $e_i-e_j$, respectively. Thus,  $\rho(Y^{ij})$ can only take the values $p_{ij}$ or $1-p_{ij}$ and $L$ can be chosen as $\max_{1 \le i < j \le N}\{p_{ij},1-p_{ij}\}$, or simpler, $1$. 

    For $v^2$, note that $(Y^{ij})^2 = [\tilde{G}]_{ij}^2(e_i e_i^T + e_j e_j^T) = ([G]_{ij}-p_{ij})^2(e_i e_i^T + e_j e_j^T)$ , which is a diagonal matrix, and further $\bbE[(Y^{ij})^2] = p_{ij}(1-p_{ij}) (e_i e_i^T + e_j e_j^T)$.  Thus the matrix 
    \begin{equation*}
        \sum_{1 \le i < j \le N}\bbE[(Y^{ij})^2]
    \end{equation*}
    is diagonal, with the $\text{ii}^{th}$ entry given by $\sum_{j \neq i} p_{ij}(1-p_{ij})$, hence $v^2 = \max_{1 \le i \le N} \sum_{j \neq i} p_{ij} (1-p_{ij})$. Alternatively, one can compute $\bbE[\tilde{G}^2]$ directly and notice that the cross-diagonal entries are exactly $0$ by the independence of the $Y^{ij}$ matrices for $1 \le i < j \le N$. 

    We now specialize to the ER and SBM cases. In the ER case, $\max_{1 \le i \le N} \sum_{j \neq i} p_{ij} (1-p_{ij})$ is independent of $i$ and equals $(N-1)p$. In the SBM case, we obtain a summation similar to the one from the proof of Lemma $\ref{lem::sbm-spectral-radius}$. Specifically, on the $i^{th}$ row we have $N/C-1$ blocks of size $C$ where the edge connection is $q$, and one block of size $N/C$ where the edge connection is $p_i$. Thus 
    \begin{equation*}
        v^2 = \max_{1 \le l \le C}  (N-N/C) q(1-q) + (N/C-1) p_l(1-p_l).
    \end{equation*}
    \end{proof}

\textbf{Dependent entry case. } So far, we focused on random network models with independent entries, i.e., the ER and SB models. Extensions to the case of dependent entries are possible as long as we can impose constraints on the covariance between entries, by adopting a modified version a Bernstein's inequality, e.g., see \cite{weak_exp_dep}.
\newpage
\section{Assignment of payoffs to edges} \label{app::undirected}

In this section, we clarify the assignment of payoff matrices $(A, B)$ to each edge. In particular, we begin with an \emph{undirected graph} $(\clN, \clE)$, in particular, one with a \emph{symmetric} adjacency matrix $G$. For each edge $(k, l) \in \clE$, we randomly assign either $A$ to $k$ and $B$ to $l$ or vice versa. 

As an example, consider a 3 player network game where the adjacency matrix is given by
$$ G = \begin{pmatrix} 0&1&0 \\ 1&0&1 \\ 0&1&0 \end{pmatrix}$$

From Assumption 1, it must be that both edges $(1, 2), (2, 3)$ are assigned the same payoff matrices $(A, B)$. One example is to assign the edge $(1, 2)$ to the payoffs $(B, A)$ and $(2, 3)$ to $(B, A)$. This yields the following payoffs for each agent
$$u_1 = x_1^\top B x_2$$
$$u_2 = x_2^\top A x_1 + x_2^\top B x_3$$
$$u_3 = x_3^\top A x_2.$$
Another choice is to assign the edge $(1, 2)$ to the payoffs $(A, B)$ and $(2, 3)$ to $(A, B)$.
$$u_1 = x_1^\top A x_2$$
$$u_2 = x_2^\top B x_1 + x_2^\top A x_3$$
$$u_3 = x_3^\top B x_2.$$
Finally, there is also the option to assign to the edge $(1, 2)$ the payoffs $(A, B)$ and $(2, 3)$ to $(B, A)$. This gives the payoffs
$$u_1 = x_1^\top A x_2$$
$$u_2 = x_2^\top B x_1 + x_2^\top B x_3$$
$$u_3 = x_3^\top A x_2.$$

Notice that, in any case, the underlying graph remains undirected, whilst the payoffs are randomly assigned and fixed throughout the game. This process can be conceptualised as randomly assigning a directed payoff matrix to each half-edge. Specifically, for each undirected edge $(k,l)$, we assign matrix $A$ to the directed half-edge $k \rightarrow l$ and matrix $B$ to the directed half-edge $l \rightarrow k$. Theorems \ref{thm::er-convergence} and \ref{thm::sbm-convergence} hold regardless of the ordering of the half-edges.
\newpage
\section{Further simulation results}
\label{appedinx_further_sim_results}

\subsection{Simulation setup}
\label{appendix:simulation_setup}
\textbf{Q-Learning hyperparameters} In all simulations, we iterate the Q-Learning algorithm defined in Equations \ref{eqn::qld_update} and \ref{eqn::boltzmann-selection}. We use learning rate $\alpha_k = 0.1$ for all agents $k$. As shown in \cite{tuyls:qlearning}, the $\alpha_k$ parameter amounts to a time rescaling in the continuous-time dynamic (\ref{eqn::QLD}), so long as all agents use the same (constant) learning rate.

\textbf{Numerical convergence} In all simulations (e.g., Figure \ref{fig::er-heatmaps}), we must evaluate numerically whether a Q-learning trajectory has converged or not. To this end, we run Q-learning for $4000$ steps and analyse the last $300$ steps. For these steps $300$, to which we will refer as the \emph{trajectory}, we compute (i) the variance of each component of the trajectory and take the mean across components, and (ii) the componentwise relative difference, defined as the maximum across components of
\begin{equation*}
    \frac{\max(\textrm{trajectory}) - \min(\textrm{trajectory})}{\max(\textrm{trajectory})}
\end{equation*}
We consider a trajectory to have converged if the mean variance is below $10^{-2}$ and relative difference below $10^{-5}$. In Figures \ref{fig::sbm-histograms} and \ref{fig::histogram_250}, we plot on the x-axis the maximum component of the absolute difference, defined for each agent $k$ as
\begin{equation*}
 \max_{i \in \clA_k} \left\{\max(\textrm{trajectory}_{ki}) - \min(\textrm{trajectory}_{ki})\right\},
\end{equation*}
where $\textrm{trajectory}_{ki}$ refers to the mixed strategy of action $i$ for agent $k$.

\textbf{Computational resources}
All experiments were run on an AMD Rome CPU cluster node with 128 cores and 2 GHz clock. 

\subsection{Additional simulations} \label{appendix:extended_simulation_study}
Figure \ref{fig::smaller_n_er} expands on Figure \ref{fig::er-heatmaps} by analysing network games with networks generated from the \er\ model, with the number of agents $N$ varying from 15 to 50 and the edge probability $p$ over a wider range $(0.05,1)$. By examining the boundary between convergent (blue) and divergent (yellow) regions, we find that convergent behaviours persist for low exploration rates only if $p$ is small. By contrast, for large $p$ (dense networks), the boundary shifts rapidly as $N$ increases. 

Figure \ref{fig::sapp-sbm-small-p} extends Figure \ref{fig::sbm-heatmaps} by illustrating the convergence of (\ref{eqn::QLD}) in Network Sato, Shapley, and Conflict games. Additionally, Figure \ref{fig::smaller_n_sbm} explores the impact of varying the number of agents from $N=15$ to $N=50$ and using the full range of $p\in(0.05,1)$. This analysis shows that both $p$ and $q$ influence the location of the boundary separating convergent and non-convergent behaviours. This finding is notable because it also applies to the Network Conflict game, which lies outside the scope of Assumption \ref{ass::bimatrix-network}

Finally, Figure \ref{fig::histogram_250} expands on Figure \ref{fig::sbm-histograms} by simulating Q-Learning on a Network Sato game with 250 agents and 1024 initial conditions. The results show that communities with low probability $p_c$ of intra-community edges exhibit less variation in the final 300 iterations compared to those with high $p_c$. This suggests that convergence is possible with lower exploration rates $T_k$ provided that heterogeneous exploration rates are employed, i.e., allowing for $T_k$ to be a function of $p_c$.

\begin{figure}[t]
    \centering
    \vspace{-1cm}
    \rotatebox{-90}{\includegraphics[width=0.95\textheight]{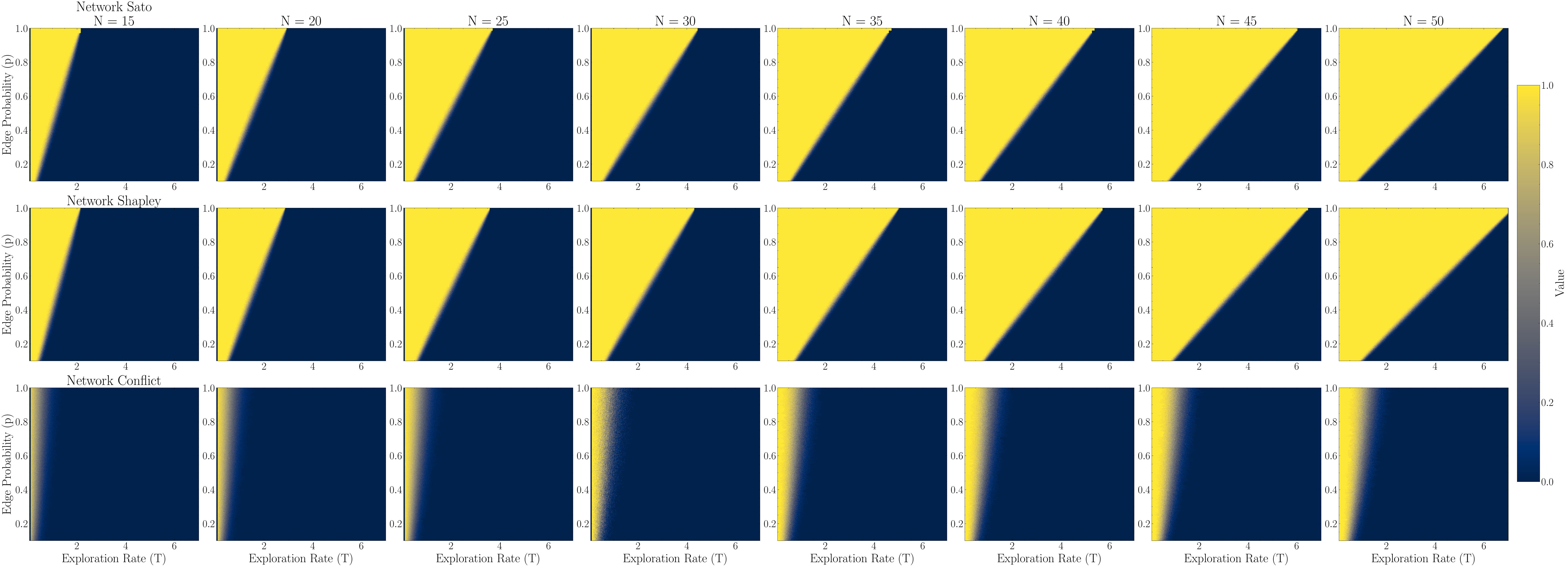}}
    \caption{Proportion of (\ref{eqn::QLD}) simulations that diverge in Network Sato, Shapley and Conflict games with networks drawn from the \er\ model, varying exploration rates $T \in (0.05,7)$, edge probability $p \in (0.05,1)$, and number of agents $N \in \{15,20,\ldots, 45,50\}$. Because we are using the full range of $p$ values up to $1$, we only display results up to $N = 50$.}
    \label{fig::smaller_n_er}
\end{figure}

\begin{figure}[t]
    \centering
        \vspace{-1cm}
    \rotatebox{-90}{\includegraphics[width=0.95\textheight]{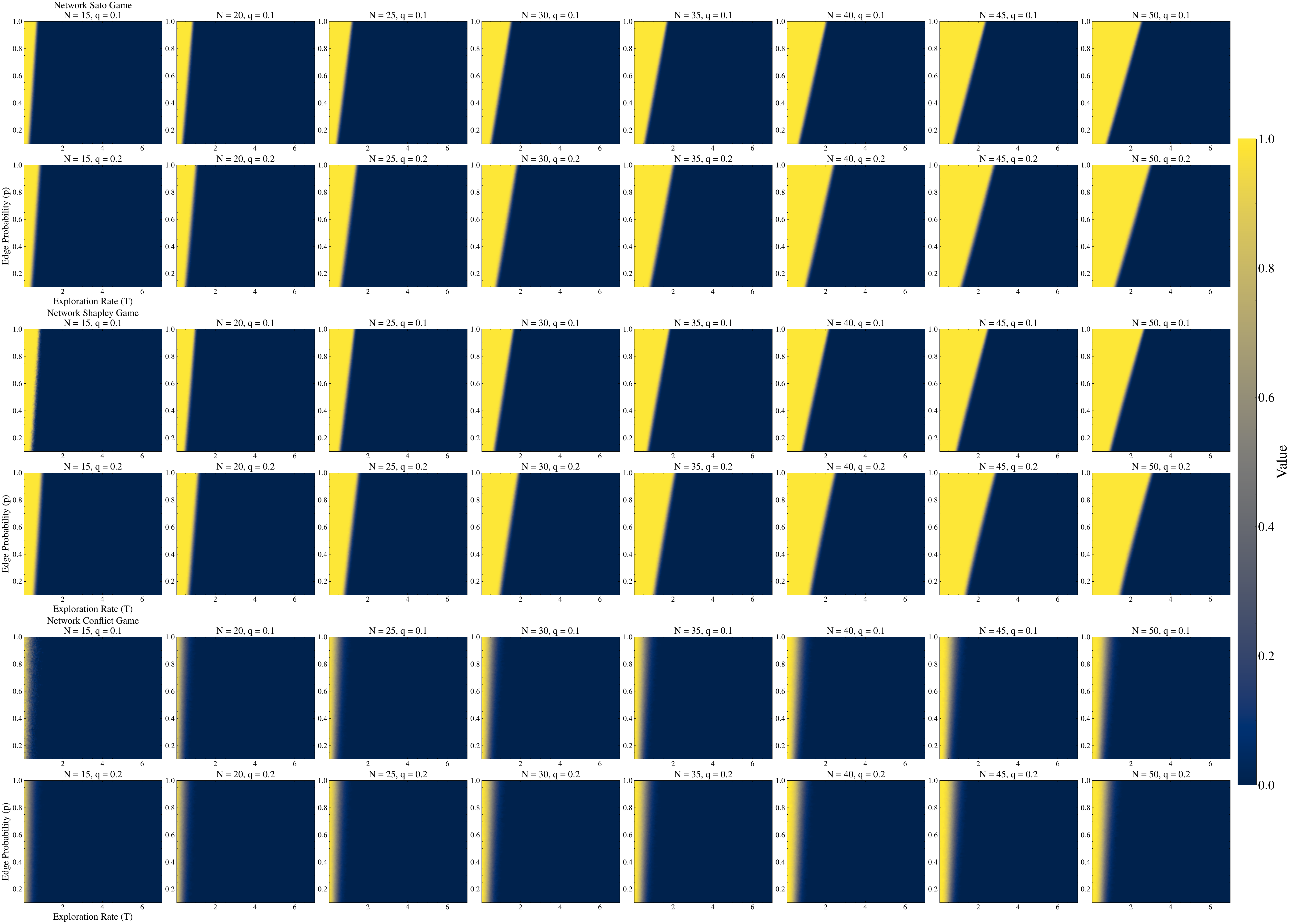}}
    \caption{Proportion of (\ref{eqn::QLD}) simulations that diverge in Network Sato, Shapley and Conflict games with networks drawn from the Stochastic Block model, varying exploration rates $T \in (0.05,7)$, intra-community edge probability $p \in (0.05,1)$, inter-community edge probability $q \in \{0.1,0.2\}$ and number of agents $N \in \{15,20,\ldots, 45, 50\}$. Because we are using the full range of $p$ values up to $1$, we only display results up to $N = 50$.}
    \label{fig::smaller_n_sbm}
\end{figure}

\begin{figure*}[t!]
    \vspace{-1cm}
     \centering
     \begin{minipage}[b]{\textwidth}
         \centering
         \includegraphics[width=0.95\textwidth]{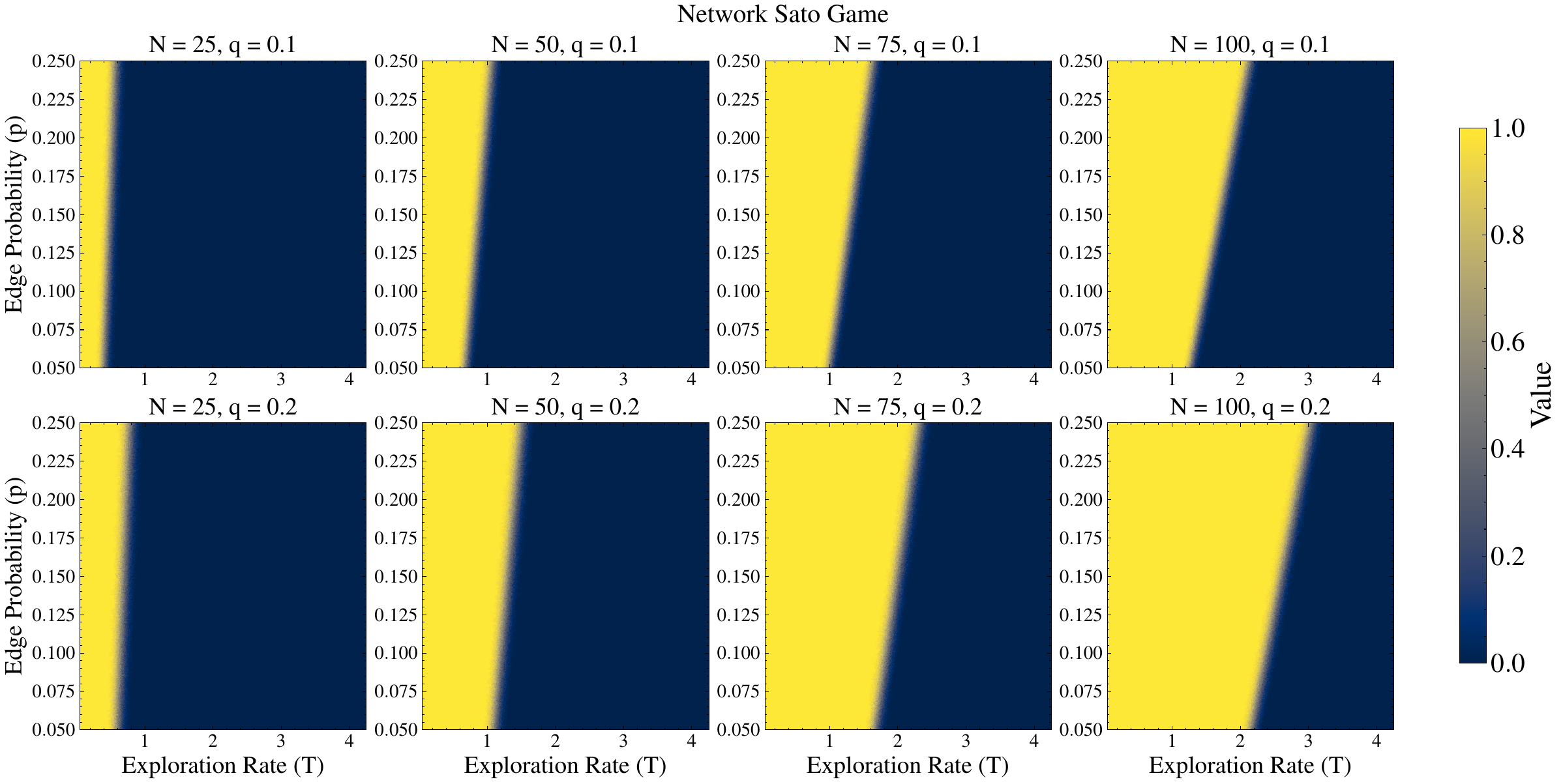}
     \end{minipage}
     \begin{minipage}[b]{\textwidth}
         \centering
         \includegraphics[width=0.95\textwidth]{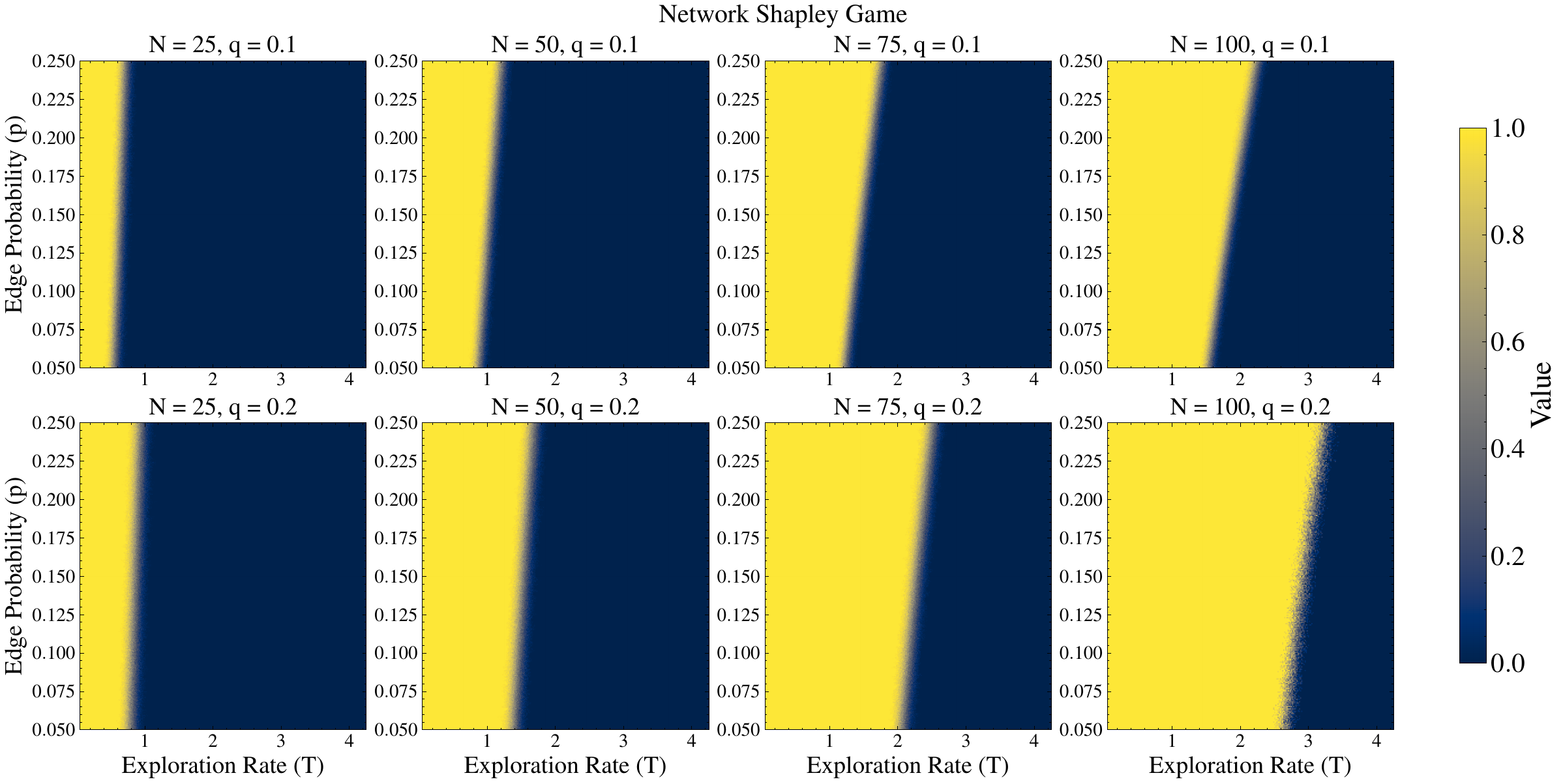}
     \end{minipage}
     \hfill
     \begin{minipage}[b]{\textwidth}
         \centering
         \includegraphics[width=0.95\textwidth]{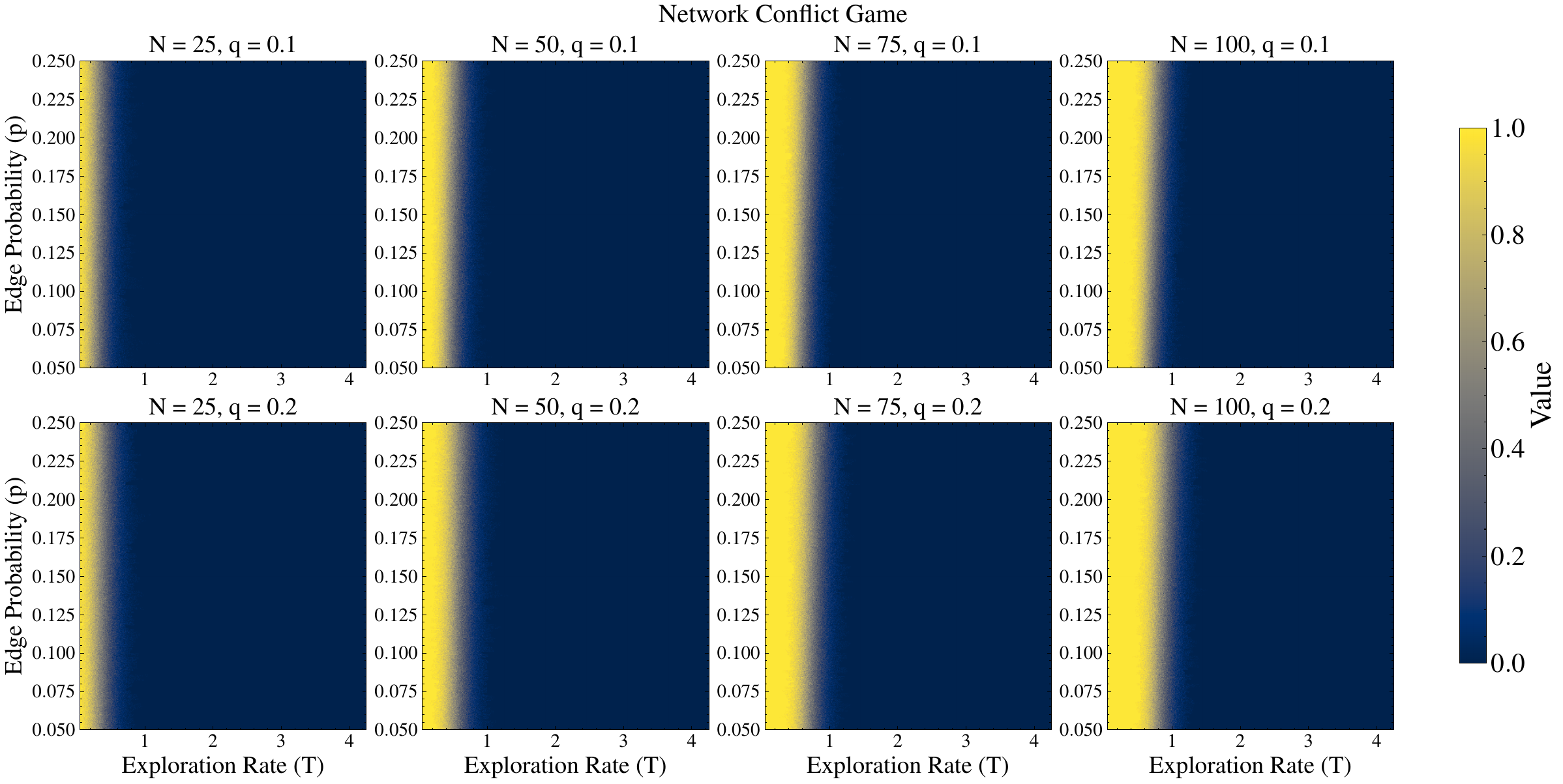}
     \end{minipage}
        \caption{Continuation of Figure \ref{fig::smaller_n_sbm} for larger number of agents $N$ and a restricted $p$ range of $(0.05,0.25)$, as used in the main body.}
        \label{fig::sapp-sbm-small-p}
\end{figure*}

\begin{figure}[t]
    \centering
    \includegraphics[width=\linewidth]{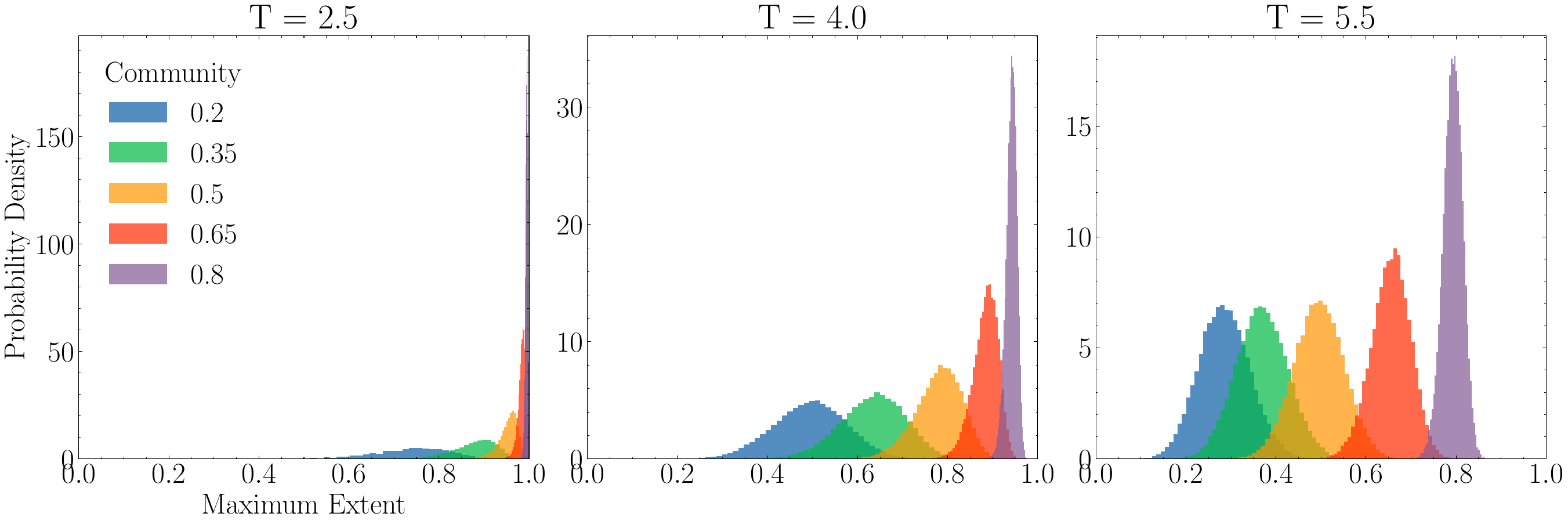}
    \caption{Probability density of final strategy variation in Network Sato games on heterogeneous stochastic block networks with $N=250$ agents, showing the maximum strategy variation across agents during the final 300 iterations, computed from 1024 independent simulations. Networks contain five communities with varying intra-community connection probabilities $p$ (shown in legend) and fixed inter-community probability $q=0.1$.}
    \label{fig::histogram_250}
\end{figure}
\clearpage
\newpage
\section{Q-Learning Dynamics}
\label{appendix:Q_learning_dynamics}
Our work is centered on independent (decentralised) multi-agent learning in normal-form games. In this setting we study the \emph{independent q-learning} algorithm \citep{leslie:iql}. Here, each agent independently maintains a Q-value of each action $i \in \clA_k$. Upon playing action $i$ at time step
$t$ and receiving its associated reward $r_{ki}$, the agent updates its Q-value to be a weighted sum
of its current estimation and the received reward. This is summarised in the update rule
\begin{equation}
    Q_{ki}(t + 1) = (1 - \alpha_k) Q_{ki}(t) + \alpha_k R_{ki}(t),
\end{equation}
where $\alpha_k \in [0, 1]$ is a step size and $R_{ki}(t)$ refers to the random realisation of the payoff at time $t$. Notice that this realisation depends on the actions of other agents in the environment, although these actions are unknown to agent $k$. 

Next, the agent updates their strategy according to their exploration policy using the \emph{Boltzmann exploration} scheme, so that
\begin{equation}
    x_{ki}(t) = \frac{\exp(Q_{ki}(t) / T_k)}{\sum_{j \in \clA{k}} \exp(Q_{kj}(t) / T_k)}
\end{equation}
By modifying the exploration rate $T_k$, the agent can smoothly move from high exploration, where
the Q-values have little influence, to high exploitation. In this manner, Q-Learning presents a strong model whereby the influence of exploration can be examined.

In \citet{sato:qlearning} and \citet{tuyls:faq}, the authors apply tools from evolutionary game theory (EGT) \citep{hofbauer:egd} to study multi-agent learning algorithms. To describe this approach, we first make a slight abuse of notation in which time $t$ will be regarded as a continuous time variable rather than discrete. Then, we can consider a learning algorithm to be a map from the joint strategy $\bfx(t)$ at time $t$ to $\bfx(t + \delta t)$ where $\delta t << 1$ is the time between consecutive updates. We will also take a deterministic approximation of the Q-update, by expressing it through the \emph{expected} reward to agent $k$, given the opponents' joint state $\bfx_{-k}(t)$. This yields the update
\begin{equation}
    Q_{ki}(t + 1) = (1 - \alpha_k) Q_{ki}(t) + \alpha_k r_{ki}(\bfx_{-k}(t)).
\end{equation}

If we take the limit as $\delta t \rightarrow 0$, we arrive at a continuous-time dynamical system that approximates the expected behaviour of the algorithm. As an example,
we depict in Figure \ref{fig::tuyls-example} traces of the Q-Learning update overlaid on its
continuous time approximation, which we introduce subsequently. The advantage of this approach is
that the tools of continuous dynamical systems can be used to prove properties of the idealisation,
which in turn tells us about the algorithm itself. For the full derivation of this dynamic from the Q-Learning update see Appendix A.1 of \citet{piliouras:zerosum} or Section 3.2 of \citet{tuyls:qlearning}.

\begin{figure}[t!]
	\centering
	\begin{minipage}[b]{0.45\textwidth}
		\includegraphics[width=\textwidth]{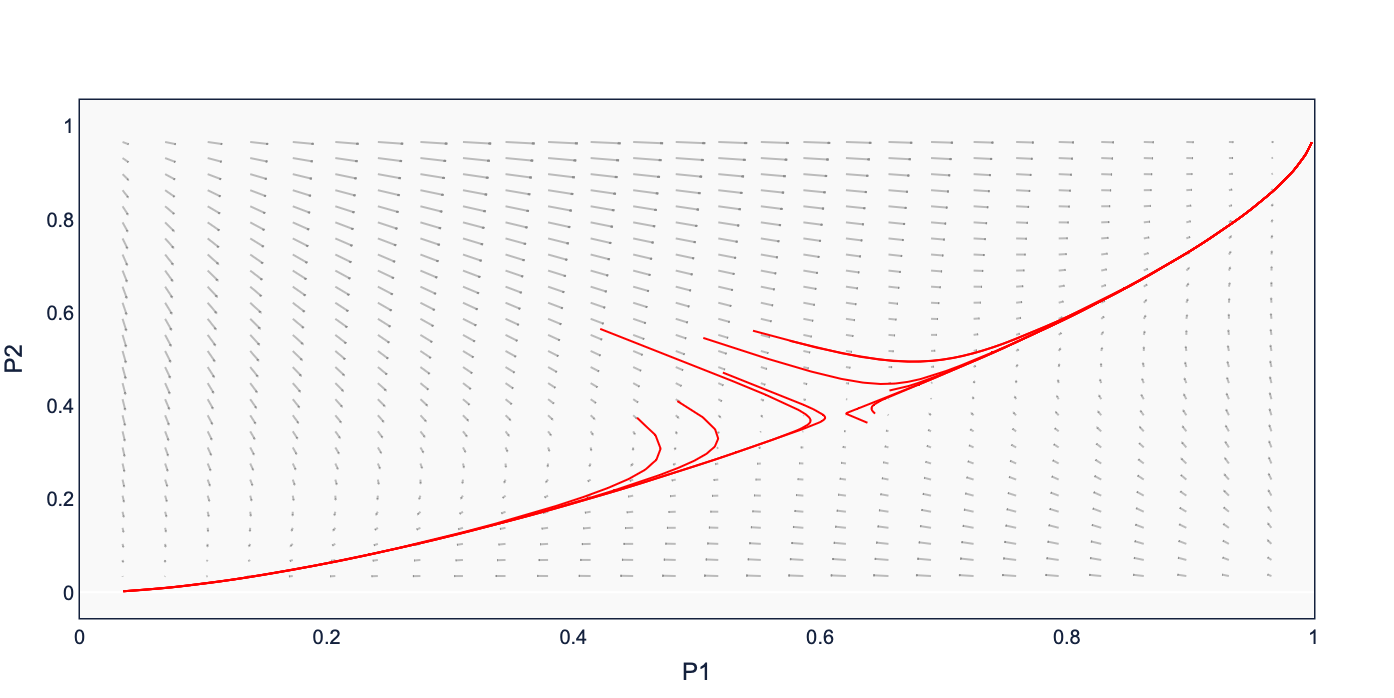}
	\end{minipage}
	\begin{minipage}[b]{0.45\textwidth}
		\includegraphics[width=\textwidth]{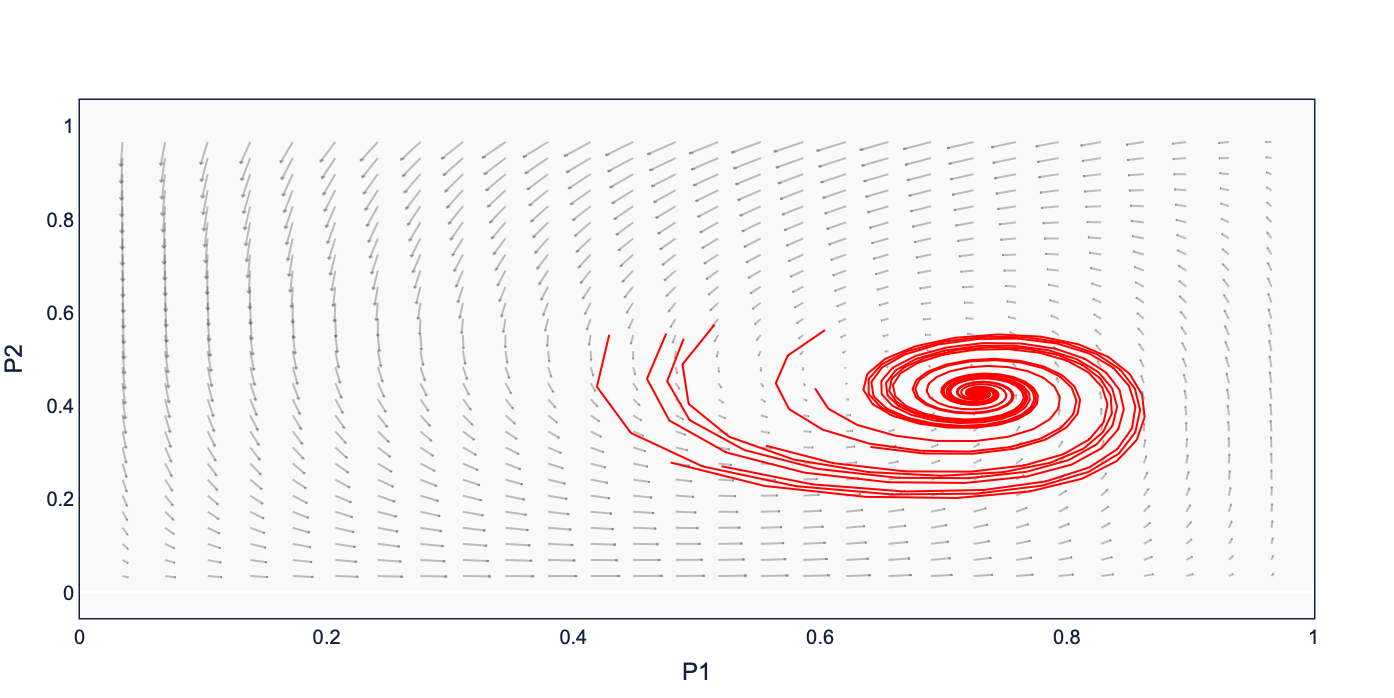}
	\end{minipage}
	\caption[Dynamics of Q-Learning]{Trajectories of the Q-Learning Algorithm (red) plotted on top of the phase portrait
    (black) generated by (\ref{eqn::QLD}) in two-player - two action games. The x-axis (resp.
    y-axis) denotes the probability with which the first (resp. second) player chooses their first
    action. In
    both cases, $T=0.3$. The payoff matrices for the left and right image are given in Figures 12 and 14 of \citet{tuyls:qlearning} respectively.} \label{fig::tuyls-example}
\end{figure}

\end{document}